\documentclass[12pt]{amsart}

\usepackage[utf8]{inputenc} 
\usepackage[T1]{fontenc}    
\usepackage{hyperref}       
\usepackage{url}            
\usepackage{booktabs}       
\usepackage{nicefrac}       
\usepackage{microtype}      

\usepackage{comment}
\usepackage{graphicx, caption}
\usepackage{pgf}
\usepackage{tikz}
\usetikzlibrary{arrows,automata}
\usepackage{enumerate}
\usepackage{caption,subcaption}
\usepackage{tabto}

\usepackage{amsmath,amsfonts,amssymb,bbm}
\usepackage[numbers,sort&compress]{natbib} 
\usepackage{bm}
\usepackage{natbib}
\usepackage{amsthm}

\usepackage{booktabs}
\usepackage[margin = 1in]{geometry}
\usepackage[utf8]{inputenc}
\usepackage{ifthen}
\usepackage[boxed, linesnumbered]{algorithm2e}

\usepackage{multicol}

\DeclareMathOperator{\Reg}{\mathbf{SSC}}
\DeclareMathOperator{\dt}{\text{dt}}
\DeclareMathOperator{\ds}{\text{ds}}

\newcommand{\lic}{\lim\limits_{{n \to \infty}}}

\newtheorem{theorem}{Theorem}[section]
\newtheorem{corollary}[theorem]{Corollary}
\newtheorem{lemma}[theorem]{Lemma}

\newtheorem{definition}[theorem]{Definition}

\newtheorem{proposition}[theorem]{Proposition}

\newtheorem{example}[theorem]{Example}

\title[Follow-the-Regularized-Leader Routes to Chaos in Routing Games]{
{Follow-the-Regularized-Leader Routes to Chaos in Routing Games}
}

\author[J. Bielawski]{Jakub Bielawski}
\address[J. Bielawski]{Department of Mathematics, Cracow University
  of Economics, Ra\-ko\-wicka~27, 31-510 Krak\'ow, Poland}
\email{jakub.bielawski@uek.krakow.pl}

\author[T. Chotibut]{Thiparat Chotibut}
\address[T. Chotibut]{
Chula Intelligent and Complex Systems, Department of Physics, Faculty of Science, Chulalongkorn University, Bangkok 10330, Thailand.}
\email{Thiparat.C@chula.ac.th, thiparatc@gmail.com}

\author[F. Falniowski]{Fryderyk Falniowski}
\address[F. Falniowski]{Department of Mathematics, Cracow University
  of Economics, Ra\-ko\-wicka~27, 31-510 Krak\'ow, Poland}
\email{falniowf@uek.krakow.pl}

\author[G. Kosiorowski]{Grzegorz Kosiorowski}
\address[G. Kosiorowski]{Department of Mathematics, Cracow University
  of Economics, Ra\-ko\-wicka~27, 31-510 Krak\'ow, Poland}
\email{grzegorz.kosiorowski@uek.krakow.pl}

\author[M. Misiurewicz]{Micha{\l} Misiurewicz}
\address[M. Misiurewicz]{Department of Mathematical Sciences, Indiana
  University-Purdue University Indianapolis, 402 N. Blackford
  Street, Indianapolis, IN 46202, USA}
\email{mmisiure@math.iupui.edu}

\author[G. Piliouras]{Georgios Piliouras}
\address[G. Piliouras]{Engineering Systems and Design, Singapore
  University of Technology and Design, 8 Somapah Road, Singapore 487372}
\email{georgios@sutd.edu.sg}

\begin{document}
\maketitle

\markleft{BIELAWSKI, CHOTIBUT, FALNIOWSKI, KOSIOROWSKI, MISIUREWICZ, PILIOURAS}

\begin{abstract}

We study the emergence of chaotic behavior of Follow-the-Regularized Leader (FoReL)  dynamics in games. 
We focus on the effects of increasing the population size or the scale of costs in  congestion games, and generalize recent results on unstable, chaotic behaviors in the Multiplicative Weights Update dynamics \cite{palaiopanos2017multiplicative,Thip18,CFMP2019} to a much larger class of FoReL dynamics.
 We establish that, even in simple linear non-atomic congestion games with two parallel links and \emph{any} fixed learning rate, unless the game is fully symmetric, increasing the population size or the scale of costs causes learning dynamics to become unstable and eventually chaotic, in the sense of Li-Yorke and positive topological entropy. 
 Furthermore, we show the existence of  novel non-standard phenomena such as the coexistence of stable Nash equilibria and chaos in the same game. We also observe the simultaneous creation of a chaotic attractor as another chaotic attractor gets destroyed. 
Lastly, although FoReL dynamics can be strange and non-equilibrating, we prove that the time average still converges to an {\it exact}  equilibrium for any choice of learning rate and any scale of costs.

\end{abstract}

\section{Introduction}

We study the dynamics of online learning in a non-atomic repeated congestion game. Namely, every iteration of the game presents a population (i.e., a continuum of players) with a choice between two strategies, and imposes on them a cost which increases with the fraction of population adopting the same strategy. In each iteration, the players update their strategy accommodating for the outcomes of previous iterations. The structure of cost function here concerns that of the congestion games, which are introduced by Rosenthal \citep{rosenthal73} and are amongst the most studied classes of games. A seminal result of \cite{monderer1996fictitious} shows that congestion games are isomorphic to potential games; as such, numerous learning dynamics are known to converge to Nash equilibria \cite{Even-Dar:2005:FCS:1070432.1070541,Fotakis08,Fischer:2006:FCW:1132516.1132608,Kleinberg09multiplicativeupdates,kleinberg2011load,MS2018,
hoo}.

A prototypical class of online learning dynamics is Follow the Regularized Leader (FoReL)~\cite{shalev,hazan2016introduction}.
  FoReL algorithm includes as special cases ubiquitous meta-algorithms, such as the Multiplicative Weights Update (MWU) algorithm ~\cite{Arora05themultiplicative}.
  Under FoReL, the strategy in each iteration is chosen by minimizing the weighted (by the learning rate) sum of the total cost of all actions chosen by the players and the regularization term.
FoReL dynamics are known to achieve optimal regret guarantees (i.e., be competitive with the best fixed action with hindsight),
as long as they are executed with a highly optimized learning rate; i.e., one that is decreasing with the steepness of the online costs (inverse to the Lipschitz constant of the online cost functions) as well as decreasing with time $T$ at a rate $1/\sqrt{T}$~\cite{shalev}.   


Although precise parameter tuning is perfectly reasonable from the perspective of algorithmic design, it seems implausible from the perspective of  behavioral game theory and modeling.  For example,  experimental and econometric studies  
 based on a behavioral game theoretic learning model known as 
  Experienced Weighted Attraction (EWA), which includes MWU as a special case, have shown that agents can use much larger learning rates  than those required for the standard regret bounds  to be meaningfully applicable~\cite{EWA1,EWA2,EWA3,camerer2011behavioral}. In some sense, such a tension is to be expected because small and optimized learning rates are designed with system stability and asymptotic optimality in mind, whereas selfish agents care more about short-term rewards which result in larger learning rates and more aggressive behavioral adaptation. Interestingly, recent work on learning in games study exactly such settings of FoReL dynamics with large, fixed step-sizes, showing that vanishing and even constant regret is possible in some game settings~\cite{bailey2019fast,bailey2019finite}.
  
For congestion games, it is reasonable to expect that increased demands and thus larger daily costs should result in steeper behavioral responses, as agents become increasingly agitated at the mounting delays. However, to capture this behavior we need to move away from the standard assumption of effective scaling down of the learning rate. Then, the costs increase and allow more general models that can incorporate non-vanishing regret.     
Thus, in this regime, FoReL dynamics in congestion games cannot be reduced to standard regret based analysis~\cite{blum2006routing}, or even Lyapunov function arguments (e.g.,~\cite{panageas2019multiplicative}), and more refined and tailored arguments will be needed. 

In a series of work ~\cite{palaiopanos2017multiplicative,Thip18,CFMP2019},  the special case of MWU was analyzed under arbitrary population, demands.  In a nutshell, for any fixed learning rate, MWU becomes unstable/chaotic even in small congestion games with just two strategies/paths as long as the total demand exceeds some critical threshold, whereas for  small population sizes it is always convergent.
 \textit{Can we extend our understanding from MWU to more general FoReL dynamics? Moreover, are the results qualitatively similar showing that the dynamic is either convergent for all initial conditions or non-convergent for almost all initial conditions, or can there be more complicated behaviors depending on the choice of the regularizer of FoReL dynamics?}

\begin{figure}
\centering
\includegraphics[width=\textwidth]{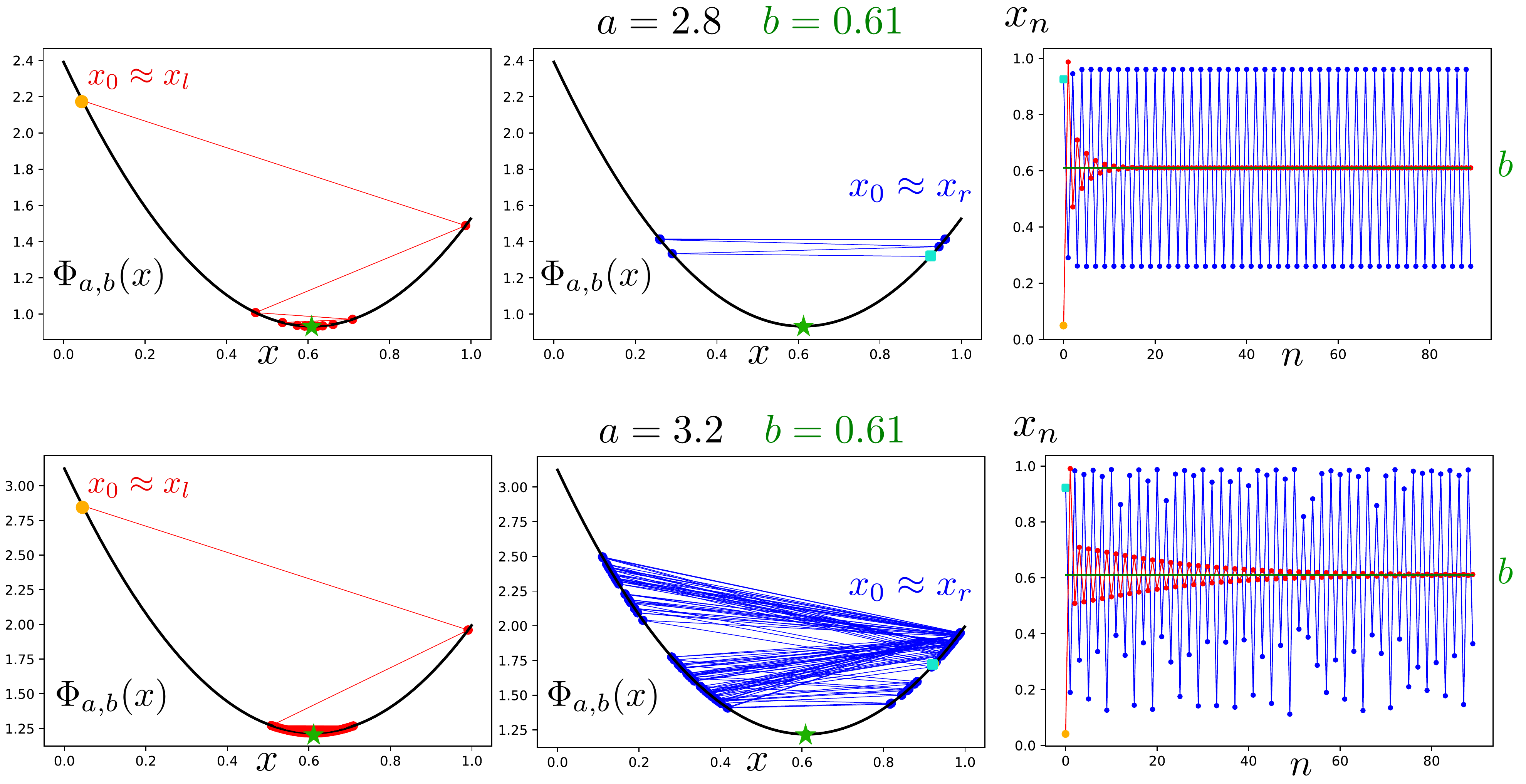}
\caption{Coexistence of locally attracting Nash equilibrium (green), limit cycles, and chaos in the same congestion game. Since congestion game has an associated convex potential (cost) function $\Phi_{a,b}(x) = \frac{a^2}{2}\left( (1-b)x^2 + b(1-x)^2\right)$ with a unique global minimum at the Nash equilibrium $b$, standard learning algorithms such as gradient-like update with a small step size will converge to the equilibrium. However, here we highlight the unusual coexistence of the attracting Nash equilibrium, limit cycles, and chaos for FoReL dynamics with  log-barrier regularizer $r(x)=(1-x)\log(1-x) + x\log(x) -\frac{5}{12}\log(-x^2+x+0.11)$. The right column shows that FoReL dynamics $x_n$ depends on the initial conditions (cyan and orange colors.) Red color encodes the dynamics initialized near the left critical point $x_l$, which converges to the Nash equilibrium $b$. Blue color encodes the dynamics initialized near the right critical point $x_r$, which converge to the limit cycle of period 2 (top), and to chaotic attractors (bottom). Convergence to the Nash equilibrium arises through dynamics that lower the cost function at every successive steps (left column), while convergence to a limit cycle or a chaotic attractor incur large cost, bouncing around in the cost landscape away from the Nash equilibrium. Remarkably, despite being periodic or chaotic, we prove that the time-average of the dynamics converges {\it exactly} to the Nash equilibrium $b$, independent of the interior initial conditions. The bifurcation diagram associated with $b=0.61$ that demonstrates coexistence of multiple attractors in the same game is shown in Fig. \ref{fig:minipathology}}
\label{fig:potential_coexistence}
\end{figure}

{\bf Our model \& results.}  We analyze FoReL-based dynamics with steep regularizers\footnote{Steepness of the regularizer guarantees that the dynamics will be well-defined as a function of the current state of the congestion game. For details, see Section \ref{s:prelim}.} in non-atomic linear congestion games with two strategies. 
This seemingly simple setting will suffice for the emergence of  highly elaborate and unpredictable behavioral patterns.
 For any such  game $G$ and an arbitrarily small but fixed learning rate $\epsilon$, we show that there exists a system capacity $N_0(G, \epsilon)$ such that the system is unstable when the total demand exceeds this threshold. In such case, we observe  \textit{complex non-equilibrating dynamics: periodic orbits of any period and chaotic behavior of trajectories} (Section \ref{s:behavior}). A core technical result is that almost all such congestion games (i.e. unless they are fully symmetric), given sufficiently large demand, will exhibit Li-Yorke chaos and positive topological entropy (Section \ref{ssec:chaos}).
  In the case of games with asymmetric equilibrium flow, the bifurcation diagram is very complex (see Section \ref{s:experimental}).
  Li-Yorke chaos implies that there exists an uncountable set of initial conditions that gets scrambled by the dynamics. Formally,  given any two initial conditions $x(0), y(0)$ in this set,  $\liminf \limits_{t\rightarrow \infty} dist(x(t),y(t))=0$ while $\limsup \limits_{t\rightarrow \infty}dist(x(t),y(t))>0$, meaning trajectories come arbitrarily close together infinitely often but also then move away again.
 In the special case where the two edges have symmetric costs (equilibrium flow is the $50-50\%$ split), the system will still become unstable given large enough demand, but chaos is not possible. Instead, in the unstable regime, all but a measure zero set of initial conditions gets attracted by periodic orbits of period two which are symmetric around the equilibrium. 
  Furthermore, we construct formal criteria for when the Nash equilibrium flow is globally attracting. For such systems we can prove their equilibration and thus social optimality  even when standard regret bounds are not applicable (Section \ref{ssec:Nash}). Also, remarkably, whether the system is equilibrating or chaotic, we prove that the time-average flows of FoReL dynamics exhibit regularity and \textit{always converge exactly to the Nash equilibrium} (Section \ref{ssec:Cesaro}). 
  
 In Section \ref{s:experimental}, for the first time, to our knowledge, we report strange dynamics arising from FoReL in congestion games. Firstly, we numerically show that for FoReL dynamics a locally attracting Nash equilibrium and chaos can coexist, see Figure \ref{fig:minipathology}. This is also formally proven in Section \ref{sec:coexistsence}. Given the prominence of local stability analysis to equilibria for numerous game theoretic settings which are widely used in Artificial Intelligence, such as Generative Adversarial Networks (GANs), e.g., \cite{goodfellow2014generative,mescheder2018training,liang2019interaction,nagarajan2017gradient,yaz2018unusual}, we believe that this result is rather important as it reveals that local stability analysis is not sufficient to guard against chaotic behaviors even in a trivial game with one (locally stable) Nash equilibrium! Secondly, Figure \ref{fig:minisimult} reveals that chaotic attractors can be non-robust. Specifically, we show that mild perturbations in the parameter can lead to the  
  destruction of one complex attractor while another totally distinct complex attractor is born!  To the best of our knowledge, these phenomena have never been reported, and thus expanding our understanding of the range of possible behaviors in game dynamics. Several more examples of complex phenomena are provided in Section \ref{s:experimental}.
Finally, further calculations for entropic regularizers can be found in Appendix \ref{sec: entropies}.
 
Our findings suggest that the chaotic behavior of players using Multiplicative Weights Update algorithm in congestion games (see results from \cite{palaiopanos2017multiplicative,Thip18,CFMP2019}) is not an exception but the rule. Chaos is robust and can be seen for a vast subclass of online learning algorithms. In particular, our results apply to an important subclass of regularizers, of generalized entropies, which are widely used concepts in information theory, complexity theory, and statistical mechanics \cite{Csiszar2008,Tsallis, SKZ}. 
Steep functions \cite{MZ,MS2016,MS2018} and generalized entropies are also often used as regularizers in game-theoretic setting \cite{CGM, MS2018, BMSS}.
In particular, Havrda-Charv\'{a}t-Tsallis entropy-based dynamics was studied, for instance, in \cite{Harper, Karev13}.
Lastly, the emergence of chaos is clearly a hardness type of result. Such results only increase in strength the simpler the class of examples is. Complicated games are harder to learn and it is harder for players to coordinate on an equilibrium. Thus, in more complicated games one should expect even more complicated, unpredictable behaviors.

\section{Model}
\label{s:prelim}
We consider a two-strategy \emph{congestion game} (see~\cite{rosenthal73}) with a continuum of players (agents), where all of the players apply the \emph{Follow the Regularized Leader} (FoReL) algorithm to update their strategies \cite{shalev}. Each of the players controls an infinitesimally small fraction of the flow.
  We assume that the total flow of all the players is equal to $N$. 
  We denote the fraction of the players adopting the first strategy at time $n$ as $x_n$. The second strategy is then chosen by $1-x_n$ fraction of the players. This model  encapsulates how a large population of commuters selects between the two alternative paths that connect the initial point to the end point.
When a large fraction of the players adopt the same
strategy, congestion arises, and the cost of choosing the same strategy increases.

{\bf Linear congestion games}:
We 
focus on linear cost functions. Specifically, the cost of each path (link, route, or strategy) is proportional to the \emph{load}. By denoting $c_j$ the cost of selecting the
strategy $j$ (when $x$ relative fraction of the agents choose the first strategy), 
\begin{equation}\label{cost}
c_1(x)=\alpha N x, \hspace{50pt} c_2(1-x)=\beta N (1-x),
\end{equation}
where $\alpha,\beta>0$ are the coefficients of proportionality. Without loss of generality we will assume throughout the paper that $\alpha+\beta=1$. Therefore, the values of $\alpha$ and $\beta=1-\alpha$ indicate how different the path costs  are from each other.

A quantity of interest is the value of the equilibrium split; i.e., the relative fraction of players using the first strategy at equilibrium. The first benefit of this formulation is that the fraction of agents using each strategy at equilibrium is independent of the flow $N$.
The second benefit is that, independent of $\alpha$, $\beta$ and $N$, playing Nash equilibrium results in the optimal social cost, which is the point of contact with the Price of Anarchy research \cite{koutsoupiad_papadimitriou, CFMP2019}.

\subsection{Learning in congestion games with FoReL algorithms}

 We assume that the players at time $n+1$ know the cost of the strategies at
time $n$ (equivalently, the realized flow (split) $(x_n, 1-x_n)$) and update
their choices according to the \emph{Follow the Regularized Leader} (FoReL) algorithm. Namely, in the period $n+1$ the players choose the first strategy with probability $x_{n+1}$ such that:

\begin{align}\label{pt1}
\begin{split}
 x_{n+1} &= \arg\min_{x\in(0,1)}\left( \varepsilon \sum_{j\leq n}\left[c_1 (x_j) \cdot x + c_2 (1-x_j) \cdot (1-x) \right] + R(x,1-x) \right)\\
  &= \arg\min_{x\in(0,1)}\left( \varepsilon \sum_{j\leq n}\left[\alpha N \cdot x_j \cdot x + \beta N \cdot (1-x_j) \cdot (1-x) \right] + R(x,1-x) \right),
  \end{split}
\end{align}
where $c_1 (x_j) \cdot x + c_2 (1-x_j) \cdot (1-x)$ is a~total cost that is inflicted on the population of agents playing against the mix $(x,1-x)$ in period $j$, while $R\colon (0,1)^2\mapsto \mathbb{R}$ is a \emph{regularizer} which represents a ``risk penalty'': namely, that term would penalize abrupt changes of strategy based on a small amount of data from previous iterations of the game. The existence of a regularizer rules out strategies that focus too much on optimizing with respect to the history of our game. 
A weight coefficient $\varepsilon>0$ of our choosing is used to balance these two terms and may be perceived as a propensity to learn and try new strategies based on new information: the larger $\varepsilon$ is, the faster the players learn and the more eager they are to update their strategies. Commonly adopted as a standard assumption, the learning rate $\varepsilon$ can be regarded as a small, fixed constant in the following analysis but its exact value is not of particular interest. Our analysis/results holds for any fixed choice of $\varepsilon$.

Note that FoReL can also be regarded as an instance of an exploration-exploitation dynamics under the multi-armed bandits framework in online learning \cite{Zhao2019}. In the limit $\varepsilon \gg 1$ such that (\ref{pt1}) is well approximated by the minimization of the cumulative expected cost 
\[\sum_{j\leq n}\left[c_1 (x_j) \cdot x + c_2 (1-x_j) \cdot (1-x) \right] = \sum_{j\leq n} c_2 (1-x_j) + \left(\sum_{j\leq n}\left[c_1(x_j) - c_2(1-x_j)\right] \right)x,\]
the minimization yields 
\[
 x_{n+1}=
    \begin{cases}
      0, & \sum_{j\leq n}\left[c_1(x_j) - c_2(1-x_j)\right] > 0, \\
      1, & \sum_{j\leq n}\left[c_1(x_j) - c_2(1-x_j)\right] < 0.
          \end{cases}
\]         
Namely, the strategy that incurs the {\it least} cumulative cost in the past time horizon is selected with probability 1. This term thus represents {\it exploitation} dynamics in reinforcement learning and multi-armed bandits framework. In the opposite limit when $\epsilon \ll 1$, (\ref{pt1}) is well approximated by the minimization of the regularizer $R(x,1-x).$ For the  Shannon entropy regularizer $R(x,1-x) = -H_S(x,1-x)= x\log x + (1-x)\log (1-x)$ that results in the Multiplicative Weight Update algorithm (see the details in Appendix \ref{sec: entropies} and Sec. \ref{sec: FORLdynamics}), its minimization  yields 
\[
 x_{n+1}= (1-x_{n+1})= 1/2.
\]   
The entropic regularization term tends to explore every strategy with equal probabilities, neglecting the information of the past cumulative cost. Thus, this regularization term corresponds to {\it exploration} dynamics.
Therefore, $\varepsilon$ adjusts the tradeoff between {\it exploration} and {\it exploitation}. The continuous time variant of  (\ref{pt1}) with the Shannon entropy regularizer has been studied as models of collective adaption \cite{sato2003,sato2005}, also known as Boltzmann $Q$ learning \cite{kianercy2012}, in which the exploitation term is interpreted as behavioral {\it adaptation} whereas the exploration term represents {\it memory loss}. More recent continuous-time variants study generalized entropies as regularizers, leading to a larger class of dynamics called Escort Replicator Dynamics \cite{Harper} which was analyzed extensively in \cite{MS2016,MS2018}.

Motivated by the continuous-time dynamics with generalized entropies, we extend FoReL discrete-time dynamics (\ref{pt1}) to a larger class of regularizers. For a given regularizer $R$, we define an auxiliary function:
\begin{equation}
r\colon (0,1)\ni x \mapsto R(x,1-x) \in \mathbb{R}.
\end{equation}
We restrict the analysis to a FoReL class of regularizers for which the dynamics implied by the algorithm is well-defined. Henceforth, we assume that $R$ is a steep symmetric convex regularizer, namely $R\in \Reg$, where:
\[ \hspace{-0.2in} \label{reg} \Reg=\left\{R\in \mathcal{C}^2((0,1)^2):\ \forall_{(x,y)\in (0,1)^2} R(y,x)=R(x,y);\ \forall_{x\in(0,1)} r''(x)>0;\ \lim_{x\to 0^+}r'(x)=-\infty \right\} .\]
These conditions on regularizers are not overly restrictive: the assumptions for convexity and symmetry of the regularizer are natural, and if $\lim_{x\to 0}r'(x)$ is finite, then the dynamics of $x_n$ from \eqref{pt1} will not be well-defined. 

Many well-known and widely used regularizers like (negative) Arimoto entropies (Shannon entropy, Havrda-Charv\'{a}t-Tsallis (HCT) entropies and log-barrier being most famous ones) and (negative) R\'{e}nyi entropies, under mild assumptions, belong to $\Reg$ (see Appendix \ref{sec: entropies}). A standard non-example is the square of the Euclidean norm $R(x,1-x)=x^2+(1-x)^2$. 

\section{The dynamics introduced by FoReL} \label{sec: FORLdynamics}

Let $R \in \Reg$. Assume that up to the iteration $n>0$ the trajectory $(x_0, x_1, \ldots, x_{n-1})$ was established by  \eqref{pt1}. Then
\[
 x_{n} = \arg\min_{x\in(0,1)}\left( N\varepsilon \sum_{j\leq n-1}\left[\alpha \cdot x_j \cdot x + \beta \cdot (1-x_j) \cdot (1-x) \right] + r(x) \right).
\]
First order condition yields
\[
 r'(x_n) = - N\varepsilon \sum_{j\leq n-1}\left[\alpha \cdot x_j - \beta \cdot (1-x_j) \right].
\]
We know that $r$ is convex, therefore the sufficient and necessary condition for $x_{n+1}$ to satisfy \eqref{pt1} takes form:
\begin{align}\label{diff_eq}
\begin{split}
 r'(x_{n+1}) &= - N\varepsilon \sum_{j\leq n}\left[\alpha \cdot x_j - \beta \cdot (1-x_j) \right]
 = r'(x_n) -  N\varepsilon \left[\alpha \cdot x_n - \beta \cdot (1-x_n) \right] \\
 &=  r'(x_n) - N\varepsilon \left[ x_n - \beta \right].
\end{split}
\end{align}


We define $\Psi\colon (0,1)\ni x \mapsto -r'(x)\in \mathbb{R}$. Table \ref{tab:table1} depicts functions $\Psi$ for different entropic regularizers\footnote{By substituting the negative Shannon entropy as $r$ in (4), that is $r(x)=R(x,1-x)=x\log x+(1-x)\log (1-x)$, we obtain the Multiplicative Weights Update algorithm.}. 
Before proceeding any further, we need to establish crucial properties of the function $\Psi$.

\begin{table}[h!]
  \begin{center}
    \caption{Homeomorphisms $\Psi$ for regularizers from $\Reg$.}
    \label{tab:table1}
    \begin{tabular}{c|c|c}
      \textbf{regularizer} &$r(x)$& $\Psi (x)$\\
      \hline
      Shannon & $x\log x+(1-x)\log(1-x)$& $\log \frac{1-x}{x}$\\
      \hline
Havrda-Charv\'{a}t-Tsallis, $q\in (0,1)$ & $\frac{1}{1-q}(1-x^q-(1-x)^q)$ &$\frac{q}{1-q}\left(x^{q-1}-(1-x)^{q-1}\right)$\\
      \hline
      R\'{e}nyi, $q\in (0,1)$ &$\frac{1}{q-1}\log (x^q+(1-x)^q)$ &$\frac{q}{1-q}\cdot \frac{x^{q-1}-(1-x)^{q-1}}{x^{q}+(1-x)^{q}}$\\
      \hline
      log-barrier & $-\log x -\log(1-x)$ &$ \frac{1}{x}-\frac{1}{1-x}$\\
    \end{tabular}
  \end{center}
\end{table}

\begin{proposition} \label{psiproperties}
Let $\Psi$ be a function derived from a regularizer from $\Reg$. Then
\begin{enumerate}[i)]
\item \label{3i} $\Psi (1-x)=-\Psi (x)$ for $x\in (0,1)$.
\item \label{3ii} $\Psi$ is a homeomorphism, $\lim\limits_{x\to 0^+}\Psi(x)=\infty$, and $\lim\limits_{x\to 1^-} \Psi(x)=-\infty$.
\end{enumerate}
\end{proposition}

\begin{proof}
Due to the condition $R(x,y)=R(y,x)$, we have that $\frac{\partial R}{\partial x}(x,1-x)=\frac{\partial R}{\partial y}(1-x,x)$. Thus, if $\varphi(x)=1-x$, then:
\begin{align*}
 \Psi(1-x) &= -[r(\varphi(x))]' = -\frac{\partial R}{\partial x}(1-x,x)+\frac{\partial R}{\partial y}(1-x,x) \\
 &= -\frac{\partial R}{\partial y}(x,1-x)+\frac{\partial R}{\partial x}(x,1-x) = r'(x) = -\Psi(x).
\end{align*}
This implies  \eqref{3i}.
Moreover  $\Psi'(x)=-r''(x)<0$. Thus, $\Psi$ is decreasing.
 \[\lim_{x\to 0^+}\Psi(x)=-\lim_{x\to 0^+}r'(x)=\infty.\]
From \eqref{3i} we obtain that $\lim_{x\to 1^-}\Psi(x)=-\infty$.
\end{proof}

By Proposition \ref{psiproperties}.\ref{3ii}, $\Psi$ is a homeomorphism between $(0,1)$ and $\mathbb{R}$. 

After substituting \begin{equation}\label{ab} a=N\varepsilon, \;\; b=\beta
\end{equation} we obtain from \eqref{diff_eq} a~general formula for the dynamics
\begin{equation} \label{eq:dyn1}
x_{n+1}=\Psi^{-1}\left( \Psi(x_n) + a(x_n-b) \right),
\end{equation}
where $a>0, b\in (0,1)$.
Thus, we introduce $f_{a,b}\colon [0,1] \mapsto [0,1]$ as 

\begin{equation} \label{eq:dyn2}
f_{a,b}(x)=\begin{cases} 0,\ \ x=0 \\ \Psi^{-1}\left( \Psi(x) + a(x-b) \right),\ \ x\in (0,1) \\ 1,\ \ x=1  \end{cases}.
\end{equation}

By the properties of $\Psi$, $f_{a,b}\colon [0,1] \mapsto [0,1]$ is continuous, and \eqref{eq:dyn2} defines a discrete dynamical system emerging from the FoReL algorithm for the pair of parameters $(a,b)$.

\begin{lemma} \label{lem33}
The following properties hold:
\begin{enumerate}[i)]
\item $f_{a,b}(x)>x$ if and only if $x<b$ and $f_{a,b}(x)<x$ if and only if $x>b$.
\item \label{commute} If $\varphi \colon (0,1)\mapsto (0,1)$ is given by $\varphi(x)=1-x$, then
$\varphi \circ f_{a,b} = f_{a,1-b} \circ \varphi.$
\item \label{generalinterval} Under the dynamics defined by \eqref{eq:dyn2}, there exists a closed invariant and globally attracting interval $I\subset (0,1)$.
\end{enumerate}

\end{lemma}

\begin{proof}
 
 We obtain (i) directly from \eqref{eq:dyn2} and the fact that $\Psi$ is decreasing.

 $\Psi$ is a homeomorphism, thus if $y = \Psi(x)$ for some $x\in (0,1)$, then
\[
 y = \Psi(\Psi^{-1}(y)) = -\Psi(1-\Psi^{-1}(y)).
\]
Hence,
\[
 \Psi^{-1}(-y) = 1-\Psi^{-1}(y).
\]
Now let $x \in (0,1)$. Then
\begin{align*}
 (\varphi \circ f_{a,b})(x) &= 1 - f_{a,b}(x) = 1 - \Psi^{-1}\left( \Psi(x) + a(x-b) \right)
 = \Psi^{-1}\left( -\Psi(x) - a(x-b) \right)\\
 &= \Psi^{-1}\left( \Psi(1-x) + a\big( (1-x)-(1-b) \big) \right) = (f_{a,1-b} \circ \varphi) (x),
\end{align*}
and (ii) follows.

 By (i), $f_{a,b}(x)>x$ for $x\in (0,b)$  and   $f_{a,b}(x)<x$ for $x\in (b,1)$. Therefore, there exists $0<\delta_1<\min \{b,1-b\}$ such that $|\frac{1}{2}-x|>|\frac{1}{2}-f_{a,b}(x)|$ for $x\in (0,1) \setminus (\delta_1,1-\delta_1)$. There exists also $\delta_2>0$ such that $f_{a,b}([\delta_1,1-\delta_1])\subset (\delta_2,1-\delta_2)$. Set $\delta=\min \{\delta_1, \delta_2\}$. Then, the interval $I=[\delta,1-\delta ]$ is invariant.

To complete the proof of (iii) we need to show that $I$ is attracting. Assume that $x\in (0,1)\setminus I$ is such that its $f_{a,b}$-trajectory never enters $I$. Since $\delta\leq \delta_1$, the distance between $f_{a,b}^n(x)$ and $I$ (that is, $d_I(f_{a,b}^n(x))$, where $d_I(z)=\delta-z$ for $z \in [0,\delta]$ and $d_I(z)=z-(1-\delta)$ for $z \in [1-\delta,1]$) is decreasing and $\delta<f(\delta)<1-\delta$. Sequence $d_I(f_{a,b}^n(x))$ is decreasing and bounded from below by $0$, so it is convergent to some $\epsilon\geq 0$. Therefore, the $\omega$-set of the trajectory of $x$ is a non-empty subset of $d_I^{-1}(\{ \epsilon \} )=I_{\epsilon}=\{\delta-\epsilon, 1-\delta+\epsilon \}$. However, no non-empty subset of $I_{\epsilon}$ can be invariant (and thus, can be an $\omega$-set of a trajectory), because $\delta-\epsilon \leq \delta_1$ and thus $f_{a,b}(I_{\epsilon}) \subset (\delta-\epsilon, 1-\delta+\epsilon)$, and $f_{a,b}(I_{\epsilon}) \cap I_{\epsilon}=\emptyset$. By this contradiction, such $x$ does not exist, thus $I$ is globally attracting.

\end{proof}

 \section{Average behavior --- Nash equilibrium is Ces\'{a}ro attracting}
 \label{ssec:Cesaro}
 
We start by studying asymptotic behavior by looking on the average behavior of orbits.
We will show that the orbits of our dynamics exhibit regular average behavior known as Ces\'{a}ro attraction to the Nash equilibrium $b$.
\begin{definition}
For an interval map $f$ a point $p$ is Ces\'aro attracting if there is a neighborhood $U$ of $p$ such that for every $x \in U$ the averages $\frac{1}{n} \sum_{k=0}^{n-1} f^k(x)$ converge to $p$.
\end{definition}

\begin{theorem}[Ces\'{a}ro attracting]
\label{thm:cesaro}
For every $a>0$, $b\in (0,1)$ and $x_0\in (0,1)$ we have

\begin{equation} \label{eq:Cesaro}
\lic \frac{1}{n} \sum_{k=0}^{n-1} f_{a,b}^k(x_0)=b.
\end{equation}

\end{theorem}

\begin{proof}

Fix $x_0\in (0,1)$ and let $x_k=f_{a,b}^k(x_0)$.

From (\ref{eq:dyn2}) we get by induction that

\begin{equation} \label{iter}
x_{n}=f_{a,b}(x_{n-1})=\Psi^{-1}\left( \Psi(x_0) + a\left(\sum_{k=0}^{n-1} (x_k -b) \right) \right).
\end{equation}

By  Lemma \ref{lem33}.\ref{generalinterval} there is $\delta>0$ such that there exists a closed, globally absorbing and invariant interval $I\subset (\delta, 1-\delta)$. Thus, for sufficiently large $n$

\[
\delta<x_n=\Psi^{-1}\left( \Psi(x_0) + a\left(\sum_{k=0}^{n-1} (x_k -b) \right)\right) < 1-\delta.
\]

$\Psi$ is decreasing, thus

\[
\Psi(\delta)> \Psi(x_0) + a \left(\sum_{k=0}^{n-1} (x_k -b) \right)> \Psi (1-\delta).
\]
Therefore
\[
\frac{1}{an}  \left(\Psi(\delta)-\Psi(x_0)\right)> \frac{1}{n} \sum_{k=0}^{n-1} x_k -b >\frac{1}{an}  \left(\Psi (1-\delta)-\Psi(x_0)\right),
\]
so
\[
 \left| \frac{1}{n} \sum_{k=0}^{n-1} x_k -b \right| < \frac{1}{an} \max\{ \left| \Psi(\delta)-\Psi(x_0) \right|, \left| \Psi(1-\delta)-\Psi(x_0) \right| \}.
\]
Thus, (\ref{eq:Cesaro}) follows.
\end{proof}

\begin{corollary}
 The center of mass of any periodic orbit $\{x_0, x_1, \ldots, x_{n-1}\}$ of $f_{a,b}$ in $(0,1)$, namely
$ \frac{x_0 + x_1 + \ldots + x_{n-1}}{n}$, is equal to $b$.
\end{corollary}

 Applying the Birkhoff Ergodic Theorem, we obtain:

\begin{corollary}\label{cmmeas}
For every probability measure $\mu$, invariant for $f_{a,b}$ and such
that $\mu(\{0,1\})=0$, we have
\[
\int_{[0,1]} x\; d\mu=b.
\]
\end{corollary}

In the following sections, we will show that, despite the regularity of the average trajectories which converge to the Nash equilibrium $b$, the trajectories themselves typically exhibit complex and diverse behaviors.

\section{Two definitions of chaos}
In this section we introduce two notions of chaotic behavior: Li-Yorke chaos and (positive) topological entropy.
Most definitions of chaos focus on complex behavior of trajectories, such as Li-Yorke chaos or fast growth of the number of distinguishable orbits of length $n$, detected by positivity of the topological entropy.

\begin{definition}[Li-Yorke chaos]  \label{LYchaos-def}
Let $(X,f)$ be a dynamical system and $(x,y)\in X\times X$.
We say that $(x,y)$ is a \emph{Li-Yorke pair} if
\[
\liminf_{n\to\infty} dist (f^n(x),f^n(y))=0,\;\;\text{and}\;\;
\limsup_{n\to\infty} dist (f^n(x),f^n(y))>0.
\]
A dynamical system $(X,f)$ is \emph{Li-Yorke chaotic} if there is an uncountable set $S\subset X$ (called \emph{scrambled set}) such that every pair $(x,y)$ with $x,y\in S$ and $x\neq y$ is a Li-Yorke pair.
\end{definition}

Intuitively orbits of two points from the scrambled set have to gather themselves arbitrarily close and  spring aside infinitely many times but (if $X$ is compact) it cannot happen simultaneously for each pair of points. 
Obviously the existence of a large scrambled set implies that orbits of points behave in unpredictable, complex way.

A crucial feature of the chaotic behavior of a dynamical system is also exponential growth of the number of distinguishable orbits. This happens if and only if the topological entropy of the system is positive. In fact positivity of topological entropy turned out to be an essential criterion of chaos \cite{GW93}.
This choice comes from the fact that the future of a deterministic (zero entropy) dynamical system can be predicted if its past is known (see \cite[Chapter 7]{Weiss}) and  positive entropy is related to randomness and chaos.
For every dynamical system over a compact phase space, we can define a number $h(f)\in[0,\infty]$ called the \emph{topological entropy} of transformation $f$. This quantity was first introduced by Adler, Konheim and McAndrew \cite{AKM} as the topological counterpart of a metric (and Shannon) entropy. In general, computing topological entropy is not an easy task. However, in the context of piecewise monotone interval maps, topological entropy is equal to the exponential growth rate of the minimal number of monotone subintervals for $f^n$.

\begin{theorem}[\cite{MS}]
Let $f$ be a piecewise monotone interval map and, for all $n\geq 1$, let $m_n$ be the minimal cardinality of a monotone partition for $f^n$. Then \[h(f)=\lim_{n\to \infty}\frac 1n \log m_n =\inf_{n\geq 1} \frac 1n \log m_n.\]
\end{theorem}

\section{Asymptotic stability of Nash equilibria}\label{sec:Nashchaos}
\subsection{Asymptotic stability of Nash equilibria}
 \label{ssec:Nash}

The dynamics induced by \eqref{eq:dyn2} admits three fixed points: $0$, $1$ and $b$. By Lemma \ref{lem33}.\ref{generalinterval} we know that all orbits starting from $(0,1)$ eventually fall into a globally attracting interval $I$. Thus, the points $0$ and $1$ are repelling.  When does the Nash equilibrium $b$ attract all point from $(0,1)$? First, we look 
when $b$ is an attracting and when it is a repelling fixed point. With this aim, we study the derivative of $f_{a,b}$:
\[
 f'_{a,b} (x) = \left(\Psi^{-1} \right)'\left( \Psi(x) + a(x-b) \right) \cdot \left( \Psi'(x) + a \right).
\]
Then, 
\begin{equation}\label{fa'}
 f'_{a,b} (b) = \left(\Psi^{-1} \right)'(\Psi(b)) \cdot \left( \Psi'(b) + a \right)=\frac{\Psi'(b) + a}{\Psi'(b)}.
\end{equation}

The fixed point $b$ is attracting if and only if $\left| f'_{a,b} (b) \right| < 1$, which is equivalent to the condition:
\begin{equation}\label{fa'<1}
\left| \Psi'(b) + a \right| < -\Psi'(b).
\end{equation}
Thus, the fixed point $b$ is  attracting if and only if $ a \in \left( 0, -2 \cdot \Psi'(b) \right)$ and repelling otherwise. We will  answer when $b$ is globally attracting on $(0,1)$.
First we will show the following auxiliary lemma.
\begin{lemma}\label{conj}
Let a function $g \colon I \mapsto \mathbb{R}$ be such that 
$g'''<0$. Then 
\[
 g'\left( \frac{x+y}{2} \right) > \frac{g(x)-g(y)}{x-y}
\]
for every $x,y \in I$.
\end{lemma}

\begin{proof} Without loss of generality we can assume that $x<y$. Then
\begin{align*}
 &\frac{y-x}{2} g'\left( \frac{x+y}{2} \right) - \int_{x}^{\frac{x+y}{2}} g'(t) \dt = \int_{x}^{\frac{x+y}{2}} \int_{t}^{\frac{x+y}{2}} g''(s) \ds\dt \\
 &> \int_{\frac{x+y}{2}}^{y} \int_{\frac{x+y}{2}}^{t} g''(s) \ds\dt = \int_{\frac{x+y}{2}}^{y} g'(t) \dt - \frac{y-x}{2} g'\left( \frac{x+y}{2} \right),
\end{align*}
where the inequality follows from the fact that $g''(s)$ is smaller in the latter region while the integration is over the set of the same size. Therefore,
\[
 g'\left( \frac{x+y}{2} \right) > \frac{1}{y-x} \int_{x}^{y} g'(t) \dt = \frac{g(y)-g(x)}{y-x},
\]
which completes the proof of lemma.
\end{proof}
The following theorem  answers, whether $b$ is globally attracting on $(0,1)$.

\begin{theorem} \label{trajconv}
Let $\Psi$ be a homeomorphism derived from a regularizer from $\Reg$. Suppose that $b$ is an attracting fixed point of $f_{a,b}$. 
If $\Psi'''<0$, then trajectories of all points from $(0,1)$ converge to $b$.
\end{theorem}

\begin{proof}
In order to prove this theorem it is sufficient to show that $f_{a,b}$ doesn't have periodic orbits of period 2. 

Suppose that $\{x_0, x_1\} \in (0,1)$ is a periodic orbit of $f_{a,b}$ of period 2. 
\[ \frac{x_0+x_1}{2} = b.\]
We have that $x_1 = \Psi^{-1}\left( \Psi(x_0) + a(x_0 - b) \right)$, and therefore,
\[ \Psi(x_1) = \Psi(x_0) + a(x_0 - b).\]
Thus, 
$
 \Psi(x_1) - \Psi(x_0) = -\frac{a}{2}(x_1 - x_0),
$
or equivalently
\begin{equation}\label{diffquo}
 a = -2 \cdot \frac{\Psi(x_1) - \Psi(x_0)}{x_1 - x_0}.
\end{equation}
By Lemma \ref{conj} \begin{equation} \label{cont1}\Psi'(b) > \frac{\Psi(x_1) - \Psi(x_0)}{x_1 - x_0}=-\frac a2,\end{equation} 
but the point $b$ is attracting if and only if $\Psi'(b)<-\frac a2$, which contradicts the inequality \eqref{cont1}. 
Therefore, $f$ has no periodic point of period 2.

Now, by \cite{blockcop}, Chapter VI, Proposition 1, every
trajectory of $f$ converges to a fixed point.

\end{proof}

\begin{corollary}\label{cor:attracting}
Let $\Psi'''<0$. Then  the Nash equilibrium $b$ attracts all points from the open interval $(0,1)$ if and only if $a\in (0,-2\Psi'(b))$.
\end{corollary}

Functions $\Psi$ derived from Shannon entropy, HCT entropy or log-barrier satisfy the inequality $\Psi'''<0$. Nevertheless, this additional condition is needed, because for an arbitrary $\Psi$ derived from $\Reg$ attracting orbits of any period may exist together with the attracting Nash equilibrium $b$.  In the next section we will discuss thoroughly an example of such behavior. This shows that even for the well-known class of FoReL algorithms knowledge of local behavior (even attraction) of the Nash equilibrium may not be enough to properly describe behavior of agents.

\begin{figure}
\centering
\begin{subfigure}[b]{0.9\textwidth}
\includegraphics[width=\textwidth,height=8cm]{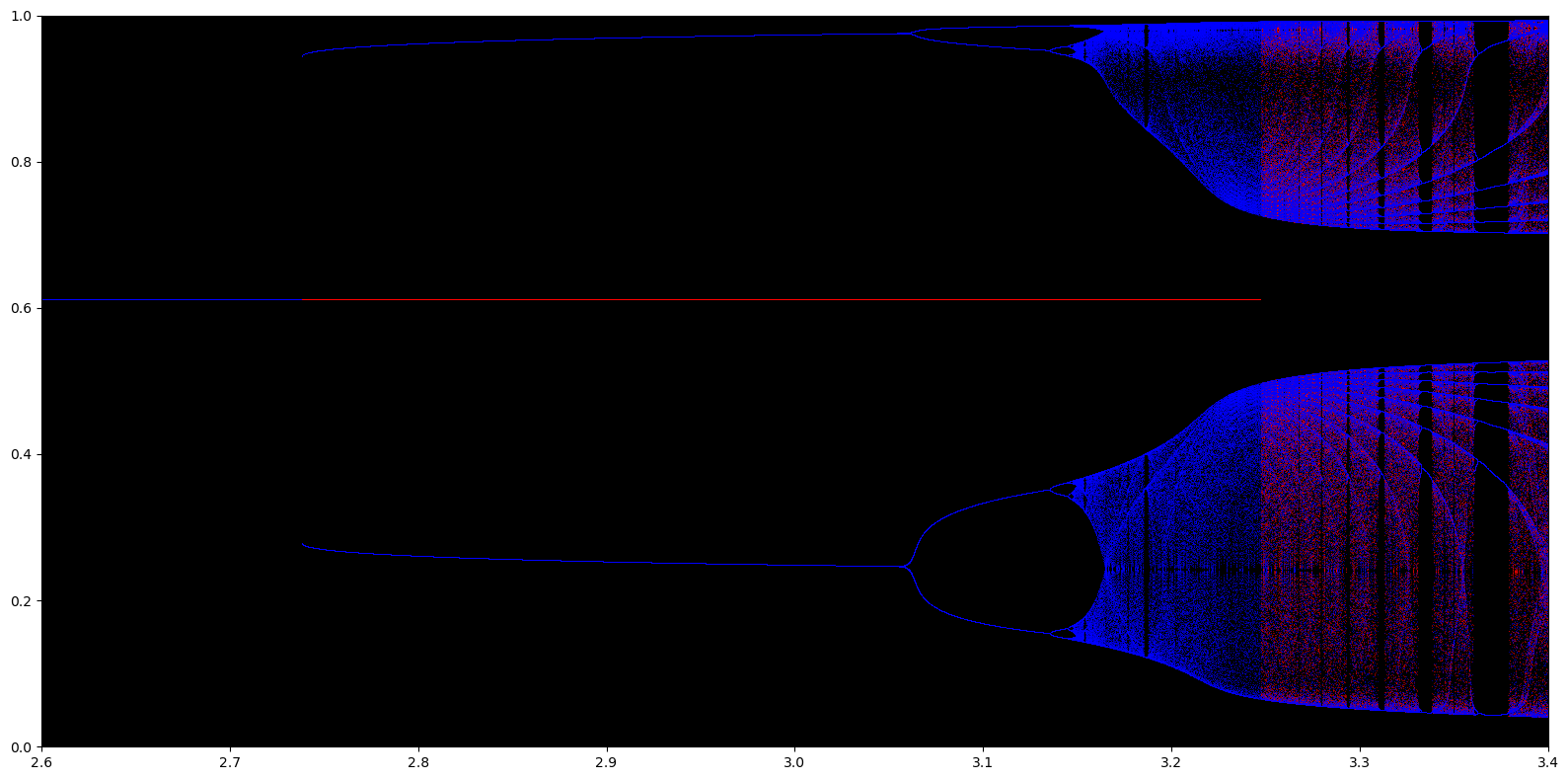}
\end{subfigure}\qquad
\begin{subfigure}[b]{.9\textwidth}
\includegraphics[width=\textwidth,height=8cm]{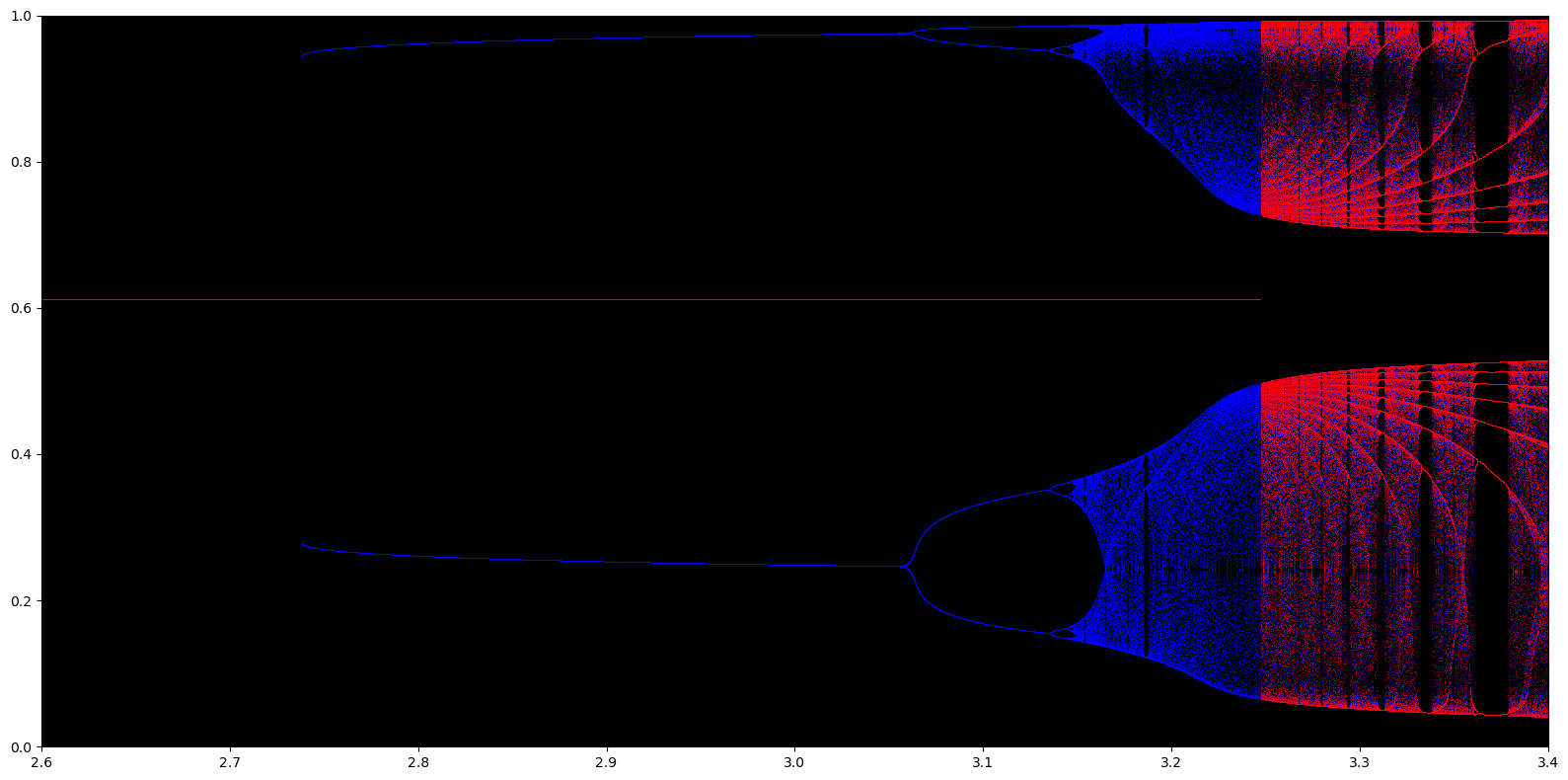}
\end{subfigure}
\caption{{\bf Coexistence of the attracting Nash equilibrium and chaos.} The bifurcation diagrams for $f_{a,b}$ where the dynamics is induced by the regularizer $r(x)=(1-x)\log (1-x)+x\log x-0.4167 \log (-x^2+x+0.11)$ for $b=0.61$. On the horizontal axis  the parameter $a$ is between $2.6$ and $3.4$, and on the vertical axis values of $f_{a,b}$ are shown.   As starting points for bifurcation diagrams two critical points of $f_{a,b}$ are taken ---  red refers to the critical point in $(0,0.5)$ and blue the critical point in $(0.5,1)$. Each critical point is iterated 4000 times, visualizing the last 200 iterates. On the top picture first red and then blue trajectories are drawn, and on the bottom one the order is reversed. Function $\Psi(x)=-r'(x)=\log (1-x)-\log x+0.4167 [\frac{1}{1.1-x}-\frac{1}{x+0.1}]$  fulfills all assumptions of Theorem \ref{trajconv} excluding $\Psi'''<0$. Although for $a<-2\Psi'(b)\approx 3.28$ the unique Nash equilibrium is attracting we can observe chaotic behavior already for $a>3.15$. The picture suggests that in the coexistence region we have an interval 
which is invariant for $f_{a,b}^2$, and in it we see the usual
evolution of unimodal maps. This means that sometimes we see an 
attracting periodic orbit, sometimes a chaotic attractor.}
\label{fig:minipathology}
\end{figure}

\subsection{Coexistence of attracting Nash equilibrium and chaos}
\label{sec:coexistsence}

In this section we will describe an example of the regularizer from
$\Reg$, which introduces game dynamics in which attracting Nash
equilibrium coexist with chaos, see example in
Figure \ref{fig:minipathology}. This phenomenon is observed
by replacing the Shannon entropic regularizer by the log-barrier
regularizer. Namely, we take
\[
\Psi(x)=\log(1-x)-\log x+0.4167\cdot \left(\frac{1}{1.1-x}
-\frac{1}{x+0.1}\right).
\]

We will show that there exist $a>0$, $b\in (0,1)$ such that $f_{a,b}$
has an attracting fixed point (which is the Nash equilibrium) yet the map can be
chaotic!

\begin{proposition}
There exist $a>0$, $b\in (0,1)$ such that $f_{a,b}$ has an attracting
fixed point (Nash equilibirum), positive topological entropy and is
Li-Yorke chaotic.
\end{proposition}

\begin{proof}
Let us take $b=0.61$ and $a=3.25$ (see the graph of $f_{a,b}$ in
Figure \ref{fig:fabnashchaos}).
\begin{figure}
\centering
\includegraphics[width=7cm,height=7cm]{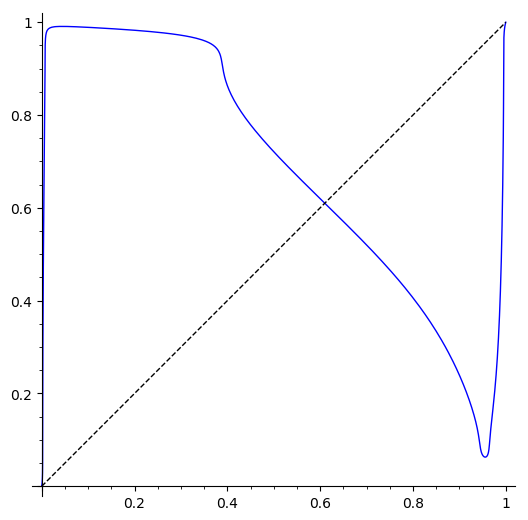}
\caption{Graph of $f_{a,b}$ for $a=3.25$, $b=0.61$ generated by the regularizer
$r(x)=(1-x)\log (1-x)+x\log x-0.4167 \log (-x^2+x+0.11)$.}
\label{fig:fabnashchaos}
\end{figure}
Set
\[
\xi (x):=\Psi(x)+a(x-b)=\log(1-x)-\log x+0.4167\cdot\left(\frac{1}
{1.1-x}-\frac{1}{x+0.1}\right)+3.25\cdot(x-0.61).
\]
Since $a<3.282596521095308\approx-2\Psi'(b)$, the fixed point $b$ is
attracting.

To show that $f_{a,b}$ is chaotic, we will prove that $f_{a,b}$ has a
periodic point of period 6. With this aim, we will show that
$(f^2)^3(x)<x<f^2(x)$ for any $x\in [0.9559,0.956]$. We start by
showing that $\xi (x)$ is monotone on $[0.9559,0.956]$. Formula for
the derivative of $\xi(x)$ is
\[
\xi'(x)=-\frac{1}{1-x}-\frac 1x+0.4167\cdot\left(\frac{1}{(1.1-x)^2}
+\frac{1}{(x+0.1)^2}\right)+3.25.
\]

Set $z=(x-0.5)^2$. Then $\xi'(x)=0$ if and only if $g(z)=0$, where
\[
g(z)=3.25\cdot z^3-1.3191\cdot z^2+0.377874\cdot z-0.050706.
\]
We have
\[
g'(z)=9.75\cdot z^2-2.6382\cdot z+0.377874,
\]
and the discriminant of this quadratic polynomial is negative.
Therefore, $g$ has only one zero (approximately $0.2077259768645677$),
so $\xi'$ has only two zeros, symmetric with respect to 0.5. Thus,
as
\[
\xi'(0.9559)\approx-0.03051955745677404,\;\;
\xi'(0.956)\approx-0.05413532613604133
\]
there is no zero of $\xi'$ between these two points.
Moreover, those computations give us an approximation to both zeros of
$\xi'$: $0.0442303467050842$ and $0.9557696532949158$.

Now we look at the first six images of $[0.9559,0.956]$:

\begin{multicols}{2}
\noindent
$\Psi(0.063)$\tabto{51pt} $\approx 0.5449390463486314$\\
$\xi(0.956) $\tabto{51pt} $ \approx 0.5450794481395858$\\
$\xi(0.9559)$\tabto{51pt} $ \approx 0.5450836794177281$\\
$\Psi(0.062)$\tabto{51pt} $ \approx 0.5458384284441133$\\

\noindent
$\Psi(0.991)$\tabto{48pt} $ \approx -1.26049734857964$\\
$\xi(0.062) $\tabto{48pt} $ \approx -1.235161571555887$\\
$\xi(0.063) $\tabto{48pt} $ \approx -1.232810953651368$\\
$\Psi(0.99) $\tabto{48pt} $ \approx -1.189231609934426$\\
\end{multicols}
\vspace{-0.8cm}

\begin{multicols}{2}
\noindent
$\Psi(0.52)$\tabto{47pt} $ \approx -0.03369120600501601$\\
$\xi(0.991)$\tabto{47pt} $ \approx -0.02224734857964017$\\
$\xi(0.99) $\tabto{47pt} $ \approx 0.04576839006557365$\\
$\Psi(0.47)$\tabto{47pt} $ \approx 0.05052025169168718$\\

\noindent
$\Psi(0.76)$\tabto{44pt} $ \approx -0.4116261583651984$\\
$\xi(0.47) $\tabto{44pt} $ \approx -0.4044797483083129$\\
$\xi(0.52) $\tabto{44pt} $ \approx -0.3261912060050159$\\
$\Psi(0.69)$\tabto{44pt} $ \approx -0.3112461911278587$\\
\end{multicols}
\vspace{-0.8cm}

\begin{multicols}{2}
\noindent
$\Psi(0.54)$\tabto{44pt} $ \approx -0.06732925721803665$\\
$\xi(0.69) $\tabto{44pt} $ \approx -0.05124619112785883$\\
$\xi(0.76) $\tabto{44pt} $ \approx 0.07587384163480165$\\
$\Psi(0.45)$\tabto{44pt} $ \approx 0.08411125490271053$\\

\noindent
$\Psi(0.8) $\tabto{44pt} $ \approx -0.4602943611198909$\\
$\xi(0.45) $\tabto{44pt} $ \approx -0.4358887450972894$\\
$\xi(0.54) $\tabto{44pt} $ \approx -0.2948292572180365$\\
$\Psi(0.65)$\tabto{44pt} $ \approx -0.2486392084062237$.
\end{multicols}
\vspace{-0.4cm}

We have $f_{a,b}(x)=\Psi^{-1}(\xi(x))$, so $\Psi(f_{a,b}(x))=\xi(x)$.
Write $\langle x,y \rangle $ for $[x,y]$ or $[y,x]$. If
$\langle \xi(x),\xi(y)\rangle \subset\langle \Psi(z),\Psi(w)\rangle$
and $\xi$ is monotone on $\langle x,y\rangle$, then $\langle
f_{a,b}(x),f_{a,b}(y)\rangle \subset\langle z,w\rangle $. Thus, the
computations show that
\[
f_{a,b}^2([0.9559,0.956])\subset[0.99,0.991]
\]
and
\[
f_{a,b}^6([0.9559,0.956])\subset[0.65,0.8].
\]
Therefore, for any $x\in[0.9559,0.956]$ we have
\[
(f_{a,b}^2)^3(x)<x<(f_{a,b}^2)(x),
\]
so by theorem from \cite{LMPY}, $f_{a,b}^2$ has a periodic point of
period 3 and $f_{a,b}$ has a periodic point of period 6. Thus, because
$f_{a,b}$ has a periodic point of period that is not a power of 2, the
topological entropy $h(f_{a,b})$ is positive (see \cite{Mis}) and it
is Li-Yorke chaotic.
\end{proof}

\begin{corollary}

There exist FoReL dynamics such that when applied to symmetric linear congestion games with only two strategies/paths the resulting dynamics have 
\begin{itemize}
\item a set of positive measure of initial conditions that converge to the unique and socially optimum Nash equilibrium and
\item 
an uncountable scrambled set for which trajectories exhibit Li-Yorke chaos,  
\item periodic orbits of all possible even periods. 
\end{itemize}
Thus,  the (long-term) social cost depends critically on the initial condition.

\end{corollary}
FoReL dynamics induced by this regularizer manifests drastically different behaviors that depend on the initial condition. 

\section{Behavior for sufficiently large $a$}\label{s:behavior}
\subsection{Non-convergence for sufficiently large $a$}
\label{ssec:chaos}

In this subsection, we study what happens as we fix $b$ and let $a$ be arbitrarily large\footnote{By \eqref{ab} it reflects the case when we fix cost functions (and learning rate $\varepsilon$) and increase the total demand $N$.}. First, we study the asymmetric case, namely $b\neq 1/2$. We show chaotic behavior of our dynamical system for $a$ sufficiently large, 
that is we will show that if $a$ is sufficiently large then $f_{a,b}$ is Li-Yorke chaotic, has periodic orbits of all periods and positive topological entropy.

The crucial ingredient of our analysis is the existence of periodic orbit of period 3.

\begin{theorem} \label{period3}
If $b \in (0,1) \backslash\{\frac 12\}$, then there exists $a_b$ such that if $a > a_b$ then $f_{a,b}$ has a periodic point of period 3.
\end{theorem}

\begin{proof}
By Lemma \ref{lem33}.\ref{commute}, without loss of generality, we may assume $b \in (0, \frac 12)$. 
We will show that  there exists $x_0 \in (0,1)$ such that $f_{a,b}^3(x_0) < x_0 < f_{a,b}(x_0)$.

Fix $a > 0$ and $b,x \in (0,1)$. We set $x_n = f_{a,b}^n(x)$, then formula \eqref{iter} holds. Hence $f_{a,b}(x) > x$ if and only if $x < b$ and, because $\Psi^{-1}$ is decreasing, $f_{a,b}^3 (x) < x$ is equivalent to $x + f_{a,b}(x) + f_{a,b}^2(x) > 3b$.

From the fact that $b\in (0,\frac 12)$ we have that $3b-1 < b$. So we can take $x_0>0$ such that $3b-1 < x_0 < b$. Then $f_{a,b}(x_0) > x_0$. Moreover
\[
 \lim_{a \to \infty} f_{a,b}(x_0) = \lim_{a \to \infty} \Psi^{-1}\left( \Psi(x_0) + a(x_0-b) \right) = 1.
\]
Thus, since $3b-x_0 < 1$, there exists $a_b >0$ such that if $a > a_b$, then $f_{a,b}(x_0) > 3b-x_0$, so $x_0 + f_{a,b}(x_0) + f_{a,b}^2(x_0) > 3b$. Hence, if $a > a_b$, then $f_{a,b}^3(x_0) < x_0$.

Now we conclude  
 that  $f_{a,b}$ has a periodic point of period 3 for $a>a_b$, from theorem from \cite{LMPY}, which implies that if $f^n(x)<x<f(x)$ for some odd $n>1$, then $f$ has a~periodic point of period $n$.

\end{proof}


By the Sharkovsky Theorem (\cite{sha}), existence of a periodic orbit of period 3 implies existence of periodic orbits of all periods, and by the result of \cite{LY}, period 3 implies Li-Yorke chaos. Moreover, because $f_{a,b}$ has a periodic point of period that is not a power of 2, the topological entropy $h(f_{a,b})$ is positive (see \cite{Mis}). 
Thus:

\begin{corollary}\label{chaos}
If $b\in (0,1)\setminus\{1/2\}$, then there exists $a_b$ such that if $a>a_b$ then $f_{a,b}$ has periodic orbits of all periods,  has positive topological entropy
and is Li-Yorke chaotic.
\end{corollary}

This result has an implication in non-atomic routing games. Recall that the parameter $a$ expresses the normalized total demand. Thus, Corollary~\ref{chaos} implies that when the costs (cost functions) of paths are different,
then increasing the total demand of the system will inevitably lead to chaotic behavior.

Now we consider the symmetric case, when $b  =\frac 12$, which
corresponds to equal coefficients of the cost functions,
$\alpha=\beta$. To simplify the notation we denote $f_a=f_{a,1/2}$.

\begin{theorem}\label{trajf}
If the parameter $a$ is small enough, then all trajectories of $f_a$ starting from $(0,1)$ converge to the attracting fixed point $1/2$.
There exists $a_b$ such that if $a>a_b$, then all points from $(0,1)$ (except countably many points, whose trajectories eventually fall into the repelling point $1/2$) are attracted by periodic attracting orbits of the form $\{\sigma_a,1-\sigma_a\}$, where $0<\sigma_a<1/2$. Moreover, if there exists $\delta>0$ such that $\Psi$ is convex on $(0,\delta)$, then there exists a unique attracting orbit $\{\sigma_a,1-\sigma_a\}$, which attracts trajectories of all points from $(0,1)$, except countably many points, whose trajectories eventually fall into the repelling fixed point $1/2$.
\end{theorem}

\begin{proof}
By~Lemma \ref{lem33}.\ref{commute}  the maps $f_a$ and $\varphi$ commute. Set $g_a=\varphi\circ
f_a=f_a\circ\varphi$. Since $\varphi$ is an involution, we have $g_a^2=f_a^2$.
We show that the dynamics of $f_a$ is simple, no matter how large $a$
is.

We aim to find fixed points and points of period 2 of $f_a$ and
$g_a$. Clearly,
\[
f_a(0)=0,\ \ \ f_a(1)=1,\ \ \ g_a(0)=1,\ \ \ g_a(1)=0.
\]

By~\eqref{iter} we have
\[
f_a^2(x)=\Psi^{-1}(\Psi(x)+a(x+f_a(x)-1)),
\]
so the fixed points of $f_a^2$ are $0$, $1$ and the solutions to
$x+f_a(x)-1=0$, that is, to $g_a(x)=x$. Thus, the fixed points of
$g_a^2$ (which, as we noticed, is equal to $f_a^2$) are the fixed
points of $g_a$ and $0$ and $1$.

We can choose the invariant interval $I_a=I_{a,1/2}$ symmetric, so
that $\varphi(I_a)=I_a$. Let us look at $G_a=g_a|_{I_a}:I_a\to I_a$. All
fixed points of $G_a^2$ are also fixed points of $G_a$, so $G_a$ has
no periodic points of period 2. By the Sharkovsky Theorem, $G_a$ has
no periodic points other than fixed points. For such maps it is known
(see, e.g.,~\cite{blockcop}) that the $\omega$-limit set of every
trajectory is a singleton of a fixed point, that is, every trajectory
converges to a fixed point. If $x\in(0,1)\setminus I_a$, then the
$g_a$-trajectory of $x$ after a finite time enters $I_a$, so
$g_a$-trajectories of all points of $(0,1)$ converge to a fixed point
of $g_a$ in $I_a$. Observe that a fixed point of $g_a$ can be a fixed
point of $f_a$ (other that 0, 1) or a periodic point of $f_a$ of
period 2. Thus, the $f_a$-trajectory of every point of $(0,1)$
converges to a fixed point or a periodic orbit of period $2$ of $f_a$,
other than $0$ and $1$.


Observe now that $1/2$ is a fixed point of both $f_a$ and $g_a$.
The fixed points of $g_a$ in $[0,1/2]$  are
the solutions of the equation $g_a(x)=x$, which
is equivalent to $f_a(x)=1-x$, further to \[\Psi(x)+a(x-1/2)=\Psi(1-x)\]
and finally, by Proposition \ref{psiproperties}.\ref{3i}, to \[2\Psi(x)=-a (x-1/2).\]
Define $\gamma_a(x)=-a/2 (x-1/2)$. We look for $\sigma_a\in (0,1/2)$ such that \begin{equation}\label{Psga} \Psi(\sigma_a)=\gamma_a(\sigma_a).\end{equation} We know that $\Psi(1/2)=\gamma_a(1/2)=0$ and $\gamma_a'(1/2)=-\frac a2$.
As $\Psi'(1/2)<0$ ($\Psi$ is strictly decreasing) there is no solution of \eqref{Psga} in $(0,1/2)$ for sufficiently small $a$. Then, $1/2$ is the only fixed point of $g_a$ in $(0,1)$. Thus, $1/2$ will attract all points from $(0,1)$.

If $\gamma_a'(1/2)<\Psi'(1/2)$, then there exists $x\in (0,1/2)$ such that $\gamma_a(x)>\Psi(x)$. Because $\lim_{x\to 0^+} \Psi(x)=+\infty$ and $\lim_{x\to 0^+} \gamma_a(x)=a/4$ and both functions are continuous, there exists $\sigma_a\in (0,1/2)$ such that $\Psi(\sigma_a)=\gamma_a(\sigma_a)$.
Finally, $\gamma_a'(1/2)<\Psi'(1/2)$ if and only if $a>-2\Psi'(1/2)$.

Lastly, if $\Psi$ is convex on some neighborhood of zero, that is $(0,\delta)$, then we can choose $a$ (sufficiently large) such that all solutions of \eqref{Psga} in $(0,\frac 12)$ lay in $(0,\delta)$. From the fact that $\Psi$ is convex on this interval and $\gamma_a$ is an affine function we obtain uniqueness of $\sigma_a$.
\end{proof}

Theorem \ref{trajf}  has a remarkable implication in non-atomic routing games. 
It implies that if cost of both paths is the same, then there is a threshold such that if the total demand will cross this threshold, then starting from almost any initial condition the system will oscillate, converging to the symmetric periodic orbit of period 2, never converging to the Nash equilibrium.
\section{Experimental results}
\label{s:experimental}

\begin{figure}
\centering
\begin{subfigure}[b]{.9\textwidth}
\includegraphics[width=\textwidth,height=8cm]{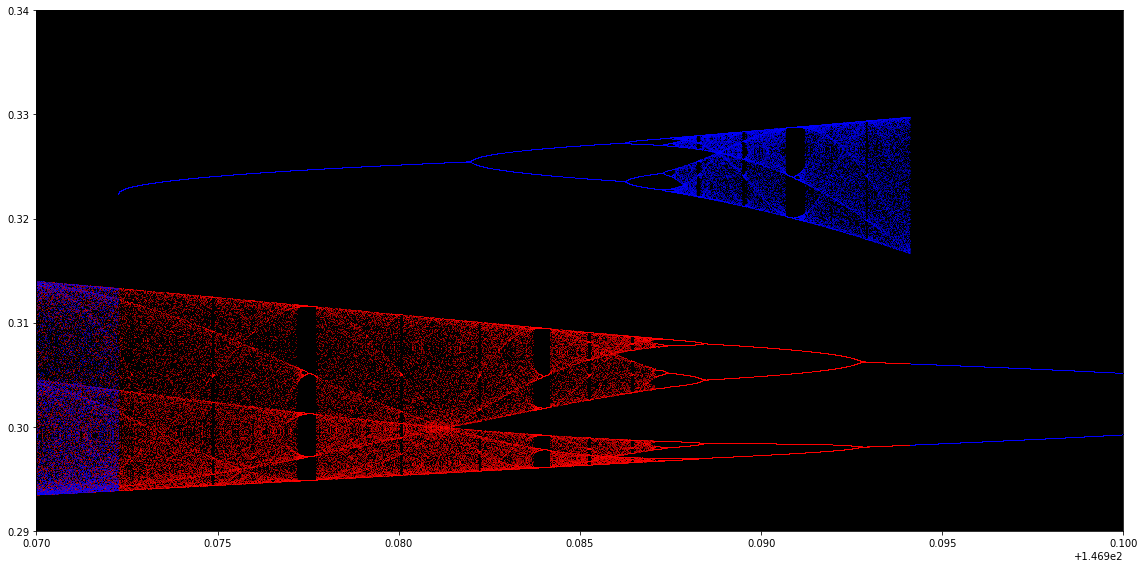}
\end{subfigure}\qquad
\begin{subfigure}[b]{.9\textwidth}
\includegraphics[width=\textwidth,height=8cm]{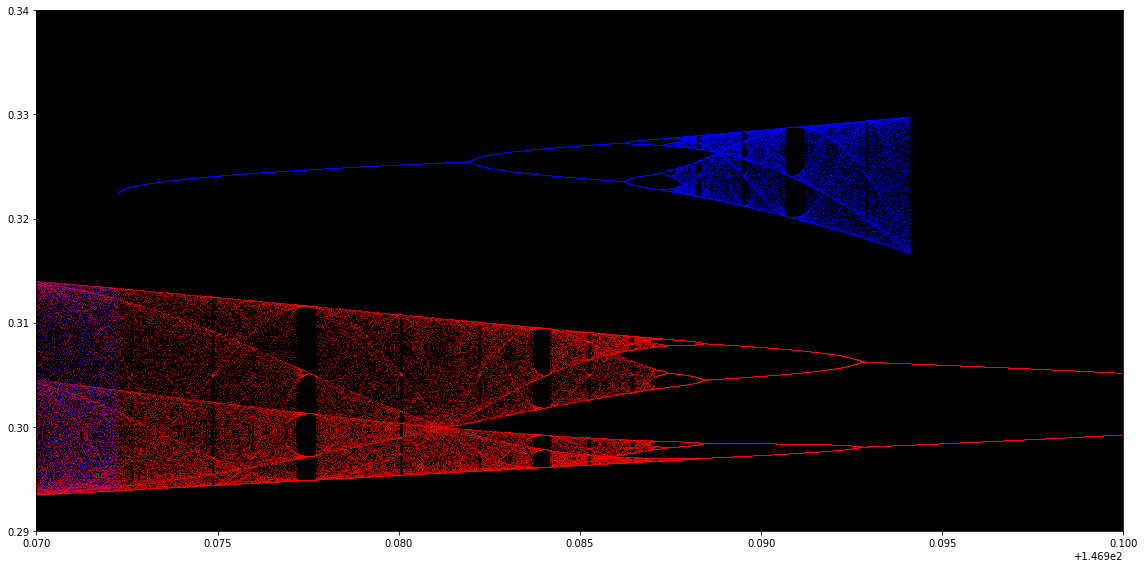}
\end{subfigure}
\caption{{\bf Simultaneous creation and destruction of different attractors.} The bifurcation diagrams for $f_{a,b}$ where the dynamics is determined by taking (negative) log-barrier regularizer with parameter $b=0.61$. On the horizontal axis  the parameter $a$ is between $146.97$ and $147$, and on the vertical axis values of $f_{a,b}$ between $0.27$ and $0.34$ are shown. 
As starting points for bifurcation diagrams two critical points of $f_{a,b}$ are taken (regularity of this map, see Appendix \ref{s:regularity}, guarantees that their trajectories detect all attractors).
 ---  red refers to the critical point in $(0,0.5)$ and blue to the critical point in $(0.5,1)$. Each critical point is iterated 4000 times, visualizing the last 200 iterates. On the top picture first red and then blue trajectories are drawn and on the bottom one  first blue and then red. We observe the collapse of the red attractor (built on the left critical point) with the simultaneous creation of the blue one (built on the right critical point).} 
  \label{fig:minisimult}
\end{figure}

In this section we report complex behaviors in bifurcation diagrams of FoReL dynamics. We investigate the structures of the attracting periodic orbits and chaotic attractors associated with the interval map $f_{a,b}: [0,1] \mapsto [0,1]$ defined by \eqref{eq:dyn2}. In the asymmetric case, that is when $b$ differs from $0.5$, the standard equilibrium analysis applies when the fixed point $b$ is stable, which is when $ |f'_{a,b}(b)| \le 1$, or equivalently when $ a \le -2\Psi'(b)$. Therefore, as we argued in the previous section, in this case the dynamics will converge toward the fixed point $b$ whenever $a < -2\Psi'(b)$. However, when $a \ge -2\Psi'(b)$ there is no attracting fixed point. Moreover, a chaotic behavior of trajectories emerges when $a$ is sufficiently large, as the period-doubling bifurcations route to chaos is guaranteed to arise.

In particular, we study the attractors of the map $f_{a,b}$ generated by the log-barrier regularizer (see Example \ref{arimoto} with $\eta(x)=\log x$) and by the Havrda-Charv\'{a}t-Tsallis regularizer for $q=0.5$ (see Example \ref{HCTent}). Note that for both of these regularizers, we have that $\Psi'''<0$. Note also that the functions $\Psi$ for these regularizers can be found in Table \ref{tab:table1}.

We first focus on the log-barrier regularizer\footnote{From the regularity of the map $f_{a,b}$ (see Appendix \ref{s:regularity}), we know that every limit cycle of the dynamics generated by $f_{a,b}$ can be found by studying the behavior of the critical points of $f_{a,b}$. Therefore, all attractors of this dynamics can be revealed by following the trajectories of these two critical points (as in Figure \ref{fig:minisimult}).}. Figure \ref{fig:minisimult} reveals an unusual bifurcation phenomenon, which, to our knowledge, is not known in other natural interval maps. We observe simultaneous evolution of two attractors in the opposite directions: one attractor, generated by the trajectory of the left critical point, is shrinking, while the other one, generated by the trajectory of the right critical point, is growing. Figure \ref{logbar153157} shows another unusual bifurcation phenomenon: a chaotic attractor arises via period-doubling bifurcations and then collapses. After that, the trajectories of the critical points, one after the other, jump, and then they together follow a period-doubling route to chaos once more.

Finally we study the bifurcation diagrams generated from Havrda-Charv\'{a}t-Tsallis regularizer with $q=0.5$. In Figure \ref{figTs397540} we observe a finite number of period-doubling (and period-halving) bifurcations, a behavior that does not lead to chaos. Nevertheless, as $a$ increases from 39.915 to 39.93, the trajectory of the right critical point leaves the attractor which it shared with the trajectory of the left critical point, and builds a separate chaotic attractor.
When chaos arises, however, we observe that the induced dynamics of the log-barrier regularizer and of the Havrda-Charv\'{a}t-Tsallis regularizer with $q=0.5$ both exhibit period-doubling routes to chaos though the regularizers are starkly different, see Figure \ref{figLB10210} and Figure \ref{figTs454} respectively.

\begin{figure}
\centering
\begin{subfigure}[b]{0.9\textwidth}
\includegraphics[width=\textwidth,height=8cm]{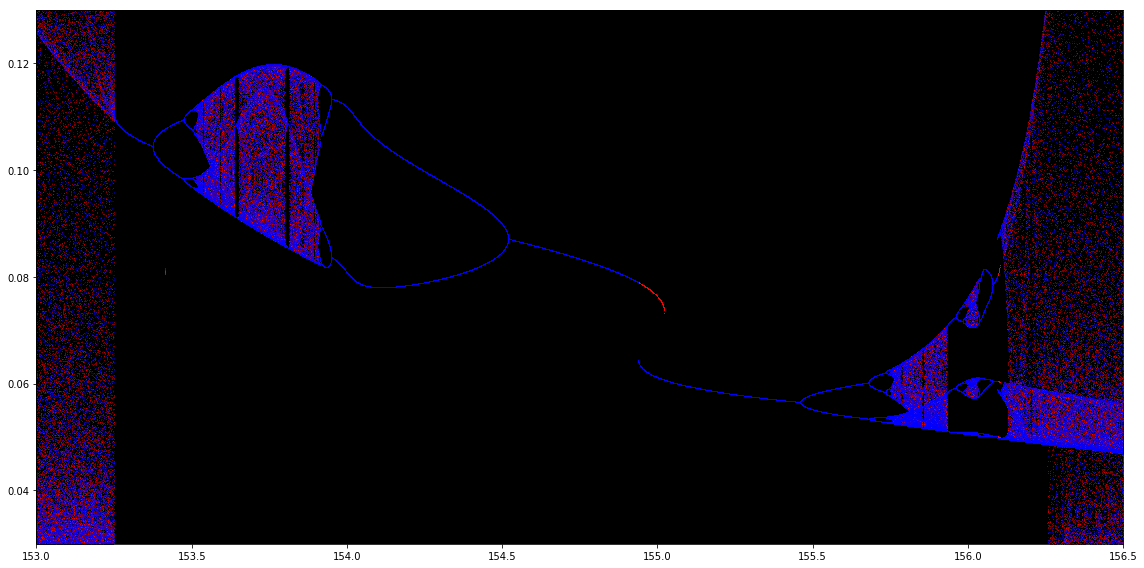}
\end{subfigure}\qquad
\begin{subfigure}[b]{.9\textwidth}
\includegraphics[width=\textwidth,height=8cm]{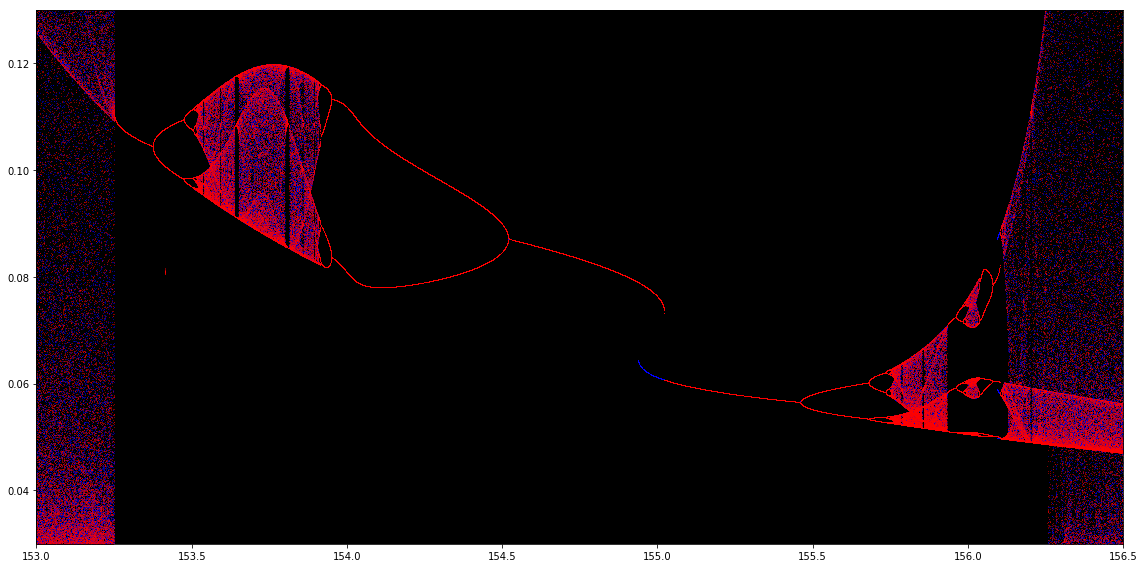}
\end{subfigure}
\caption{{\bf Locally complex behavior.} The bifurcation diagrams for $f_{a,b}$ where the dynamics is determined by taking (negative) log-barrier as the regularizer for $b=0.61$.  On the horizontal axis  the parameter $a$ is between $153$ and $156.5$, and on the vertical axis values of $f_{a,b}$ are between $0.03$ and $0.13$. As starting points for bifurcation diagrams two critical points of $f_{a,b}$ are taken 
---  red refers to the critical point in $(0,0.5)$ and blue the critical point in $(0.5,1)$. Each critical point is iterated 4000 times, then visualizing the last 200 iterates. On the top picture first red and then blue trajectories are drawn, and on the bottom one  the order is reversed. 
\noindent As $a$ increases chaotic behavior of orbits disappears (around $153.25$). Then,  within the window $[153.25,156.3]$, chaos emerges at $[153.5,154]$ and vanishes. Then trajectories jump, one after the other, and then generate a chaotic attractor which then spreads, vanishes, and finally spreads onto the whole interval.
} 
\label{logbar153157}
\end{figure}

\begin{figure}
\centering
\begin{subfigure}[b]{0.9\textwidth}
\includegraphics[width=\textwidth,height=8cm]{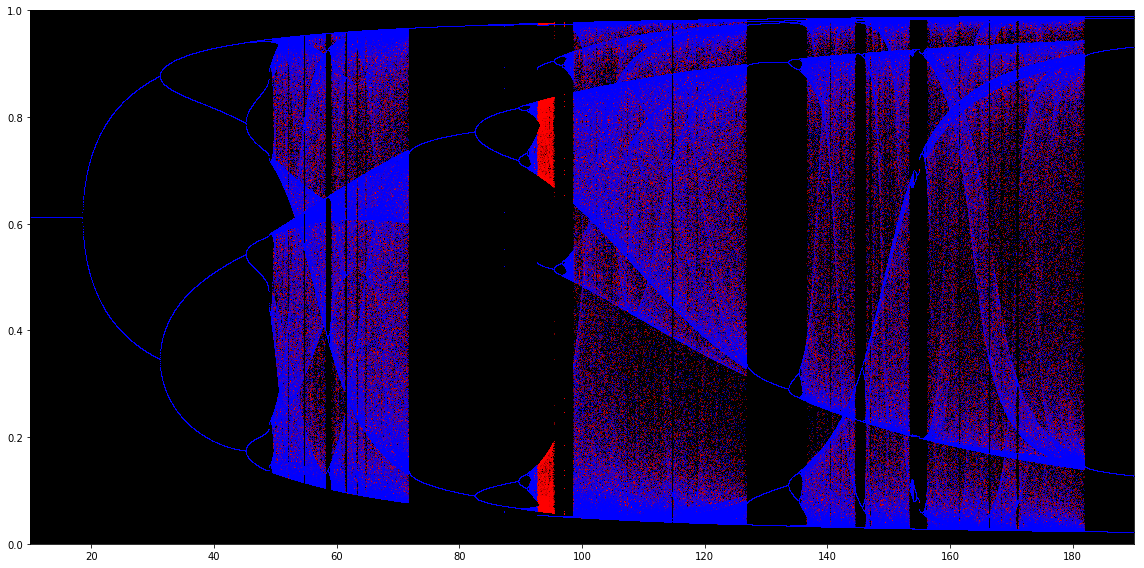}
\end{subfigure}\qquad
\begin{subfigure}[b]{.9\textwidth}
\includegraphics[width=\textwidth,height=8cm]{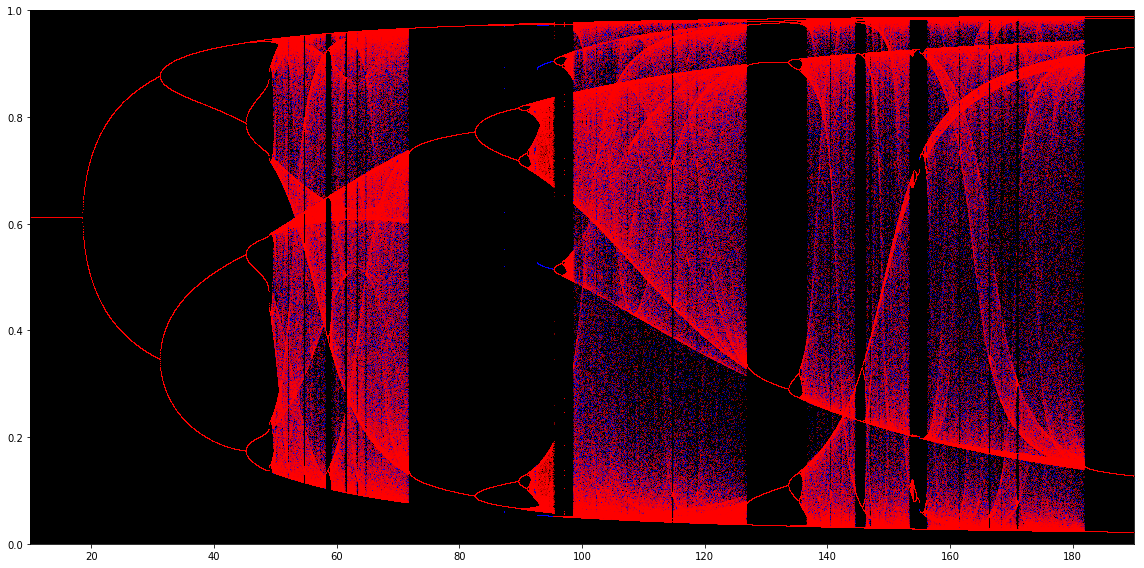}
\end{subfigure}
\caption{{\bf Period-doubling road to chaos.} The bifurcation diagrams for $f_{a,b}$ where the dynamics is determined by taking (negative) log-barrier as the regularizer: $r(x)=-\log x-\log(1-x)$ for   
$b=0.61$. On the horizontal axis  the parameter $a$ is between $10$ and $190$, and on the vertical axis the value of $f_{a,b}$ ranges between $0$ and $1$. As starting points for bifurcation diagrams two critical points of $f_{a,b}$ are taken (regularity of this map, see Appendix \ref{s:regularity}, guarantees that by studying their trajectories we visit all attractors) ---  red refers to the left critical point (in $(0,0.5)$) and blue to the right critical point (in $(0.5,1)$). Each critical point is iterated 4000 times, visualizing the last 200 iterates. On the top picture first red and then blue trajectories are drawn, and on the bottom one  the order is reversed. 
\noindent The first bifurcation takes place at the moment when the Nash equilibrium $b$ becomes repelling. Then we observe period-doubling route to chaos. In addition two different attractors are visible for $a\in (92,96)$. }
\label{figLB10210}
\end{figure}

\begin{figure}
\centering
\begin{subfigure}[b]{0.9\textwidth}
\includegraphics[width=\textwidth,height=8cm]{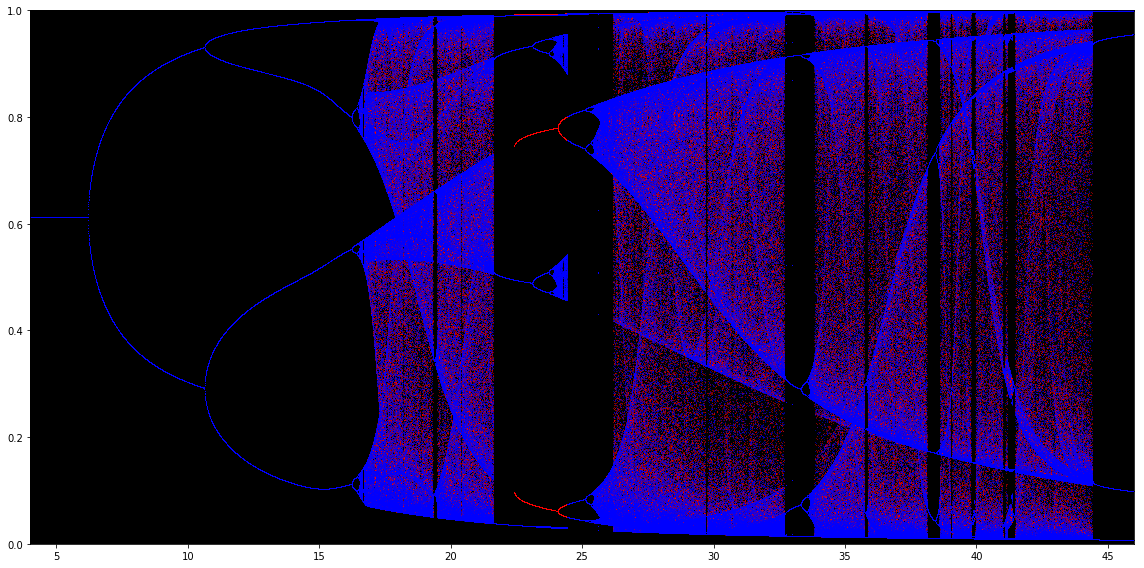}
\end{subfigure}\qquad
\begin{subfigure}[b]{.9\textwidth}
\includegraphics[width=\textwidth,height=8cm]{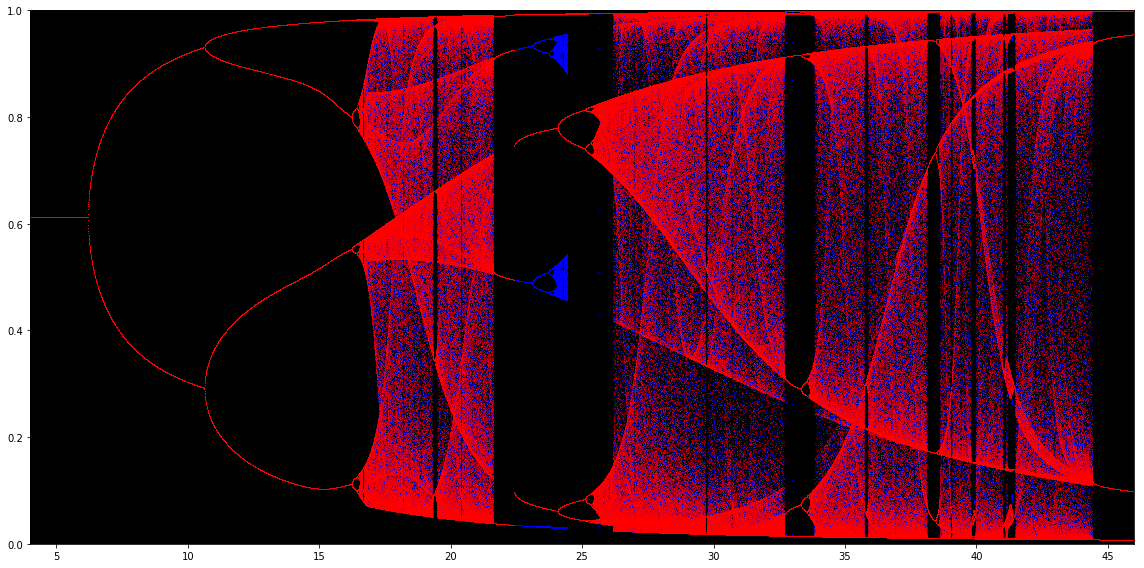}
\end{subfigure}
\caption{{\bf Period-doubling road to chaos with Havrda-Charv\'{a}t-Tsallis regularizer.} The bifurcation diagrams for $f_{a,b}$ where the dynamics is determined by taking (negative)  Havrda-Charv\'{a}t-Tsallis entropy with $q=0.5$ as the regularizer and  $b=0.61$. On the horizontal axis  the parameter $a$ is between $4$ and $46$, and on the vertical axis values of $f_{a,b}$ ranges between $0$ and $1$. As starting points for bifurcation diagrams two critical points of $f_{a,b}$ are taken ---  red refers to the critical point in $(0,0.5)$ and blue the critical point in $(0.5,1)$. Each critical point is iterated 4000 times, then visualizing the last 200 iterates. On the top picture first red and then blue trajectories are drawn, and on the bottom one  the order is reversed. 
The first bifurcation takes place at the moment when the Nash equilibrium $b$ becomes repelling. Then we observe period-doubling route to chaos. In addition two different attractors are visible for $a\in (22.5,24.5)$.} 
\label{figTs454}
\end{figure}

\begin{figure}
\centering
\begin{subfigure}[b]{0.9\textwidth}
\includegraphics[width=\textwidth,height=8cm]{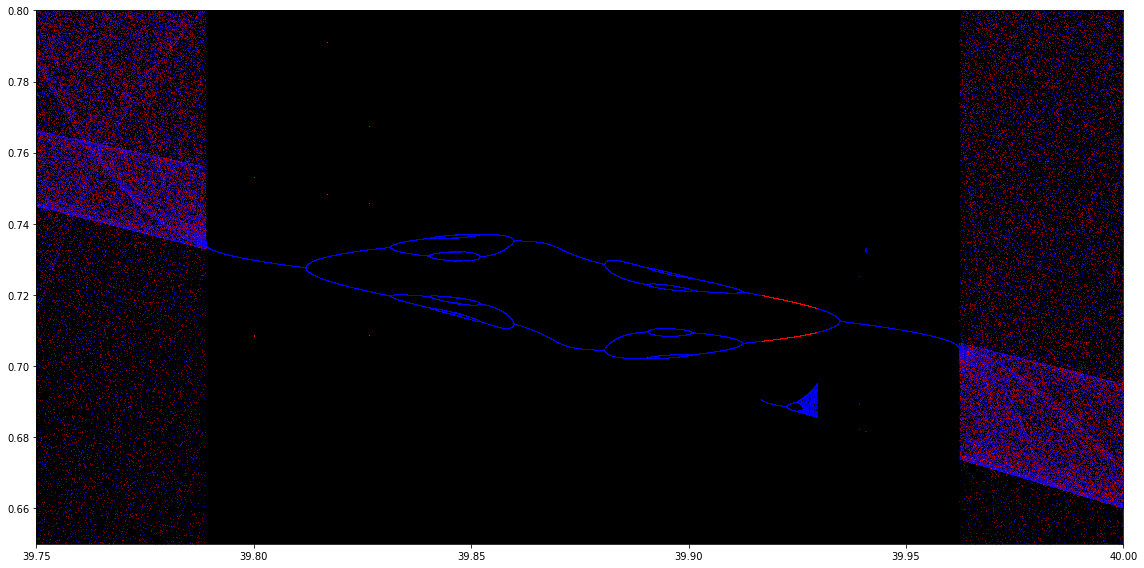}
\end{subfigure}\qquad
\begin{subfigure}[b]{.9\textwidth}
\includegraphics[width=\textwidth,height=8cm]{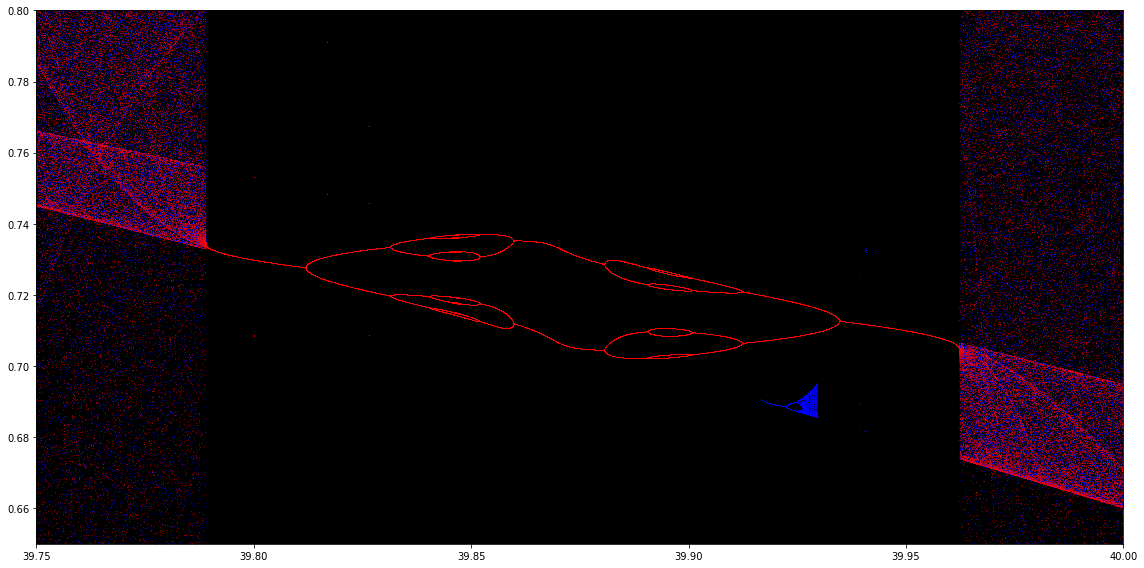}
\end{subfigure}
\caption{{\bf Period-doubling not always lead to chaos.}The bifurcation diagrams for $f_{a,b}$ where the dynamics is determined by taking (negative) Havrda-Charv\'{a}t-Tsallis entropy with $q=0.5$ as the regularizer, that is, $r(x)=\frac{1}{\sqrt{x}}-\frac{1}{\sqrt{1-x}}$. We fix $b=0.61$. On the horizontal axis  the parameter $a$ is between $39.75$ and $40$, and on the vertical axis values of $f_{a,b}$ are between $0.65$ and $0.8$. As starting points for bifurcation diagrams two critical points of $f_{a,b}$ are taken ---  red refers to the critical point in $(0,0.5)$ and blue the critical point in $(0.5,1)$. Each critical point is iterated 4000 times, then visualizing the last 200 iterates. On the top picture first red and then blue trajectories are drawn, and on the bottom one  the order is reversed. 
As $a$ increases both trajectories go through the same forward and backward period doubling steps. Then, as $a$ increases from $39.915$ to $39.93$, the trajectory of the right critical point escapes the attractor which she shared with the trajectory of the left critical point, and builds separate chaotic attractor. Then it jumps back to the red attractor.}
\label{figTs397540}
\end{figure}


\section{Conclusion}

We study FoReL dynamics in non-atomic congestion games with arbitrarily small but fixed step-sizes, rather than with decreasing and regret-optimizing step-sizes. Our model allows for agents that can learn over time (e.g., by tracking the cumulative performance of all actions to inform about their future decisions), while being driven by opportunities for    
short-term rewards, rather than only by long-term asymptotic guarantees. As a result, we can study the effects of increasing system demand and delays on agents' responses, which can become steeper as they are increasingly agitated by the increasing costs. Such assumptions are well justified from a behavioral game theory perspective~\cite{EWA1,camerer2011behavioral}; however, FoReL dynamics are pushed outside of the standard parameter regime in which classic black-box regret bounds do not apply meaningfully.
Using tools from dynamical systems, we show that, under sufficiently large demand,  dynamics will unavoidably become chaotic and unpredictable. Thus, our work vastly generalizes previous results that hold in the special case of Multiplicative Weights Update~\cite{palaiopanos2017multiplicative,Thip18,CFMP2019}. We also report a variety of undocumented complex behaviors such as the co-existence of a locally attracting Nash equilibrium and of chaos \emph{in the same game}. Despite this behavioral complexity of the day-to-day behavior,  the time-average system behavior is \emph{always} perfectly regular, converging to an exact equilibrium. 
Our analysis showcases that local stability in congestion games should not be considered as a foregone conclusion and paves the way toward further investigations at the intersection of optimization theory, (behavioral) game theory, and dynamical systems.

\section*{Acknowledgements}

Georgios Piliouras acknowledge AcRF Tier 2 grant 2016-T2-1-170, grant PIE-SGP-AI-2018-01, NRF2019-NRF- ANR095 ALIAS grant and NRF 2018 Fellowship NRF-NRFF2018-07.
Fryderyk Falniowski acknowledges the support of the
National Science Centre, Poland, grant 2016/21/D/HS4/01798 
and COST Action CA16228 ``European Network for Game Theory''.
Research of Micha{\l} Misiurewicz was partially supported by grant
number 426602 from the Simons Foundation.
Jakub Bielawski and Grzegorz Kosiorowski acknowledge support from a subsidy granted to Cracow University of Economics. Thiparat Chotibut acknowledges a fruitful discussion with Tanapat Deesuwan, and was partially supported by grants for development of new faculty staff, Ratchadaphiseksomphot endownment fund, and Sci-Super VI fund, Chulalongkorn University.

\bibliographystyle{abbrv} 
\bibliography{./FoReL_arxiv_submission.bib}

\newpage

\appendix

\label{s:appendix}

\begin{center}
	\Large{{\bf Appendices}} 
\end{center}
\section{Generalized entropies as regularizers}
\label{sec: entropies}
We present information measures which are often used as regularizers.
\begin{example}[Shannon entropy] \label{SE}
Let $R(x,y)=-H_S(x,y)$, where
\[H_S(x,y)=-x\log x-y\log y.\]

 Then $H_S(x,1-x)$ is the Shannon entropy of a probability distribution $(x,1-x)$ and
\[
r(x)=R(x,1-x)=-H_S(x,1-x)=x\log x + (1-x)\log (1-x).
\]

From \[r'(x)=\log \frac{x}{1-x}\] we observe that $r\in \Reg$.\footnote{By substituting the (negative) Shannon Entropy as $R$ into (4) we obtain the Multiplicative Weights Update algorithm.}
\end{example}

\begin{example}[Arimoto entropies] \label{arimoto}
We consider the class of Arimoto entropies \cite{Csiszar2008}, that is functions defined as \[H_{\eta}(x,y)=\eta(x)+\eta(y),\] where $\eta\in \mathcal{C}^2((0,1))$ is a concave function.\footnote{In the decision theory Arimoto entropies correspond to separable Bregman scores \cite{GD2004}.} 

We define \[R(x,y)=-H_{\eta}(x,y).\]

Then, by its definition, $R\in \mathcal{C}^2((0,1)^2)$ and $R(y,x)=R(x,y)$.
Moreover,
\[
r(x)=R(x,1-x)=-H_{\eta}(x,1-x)=-\eta(x)-\eta(1-x),
\]
\[ r'(x)=-\eta'(x)+\eta'(1-x) \text{ and } r''(x)=-\eta''(x)- \eta''(1-x).\] Thus, $r$ is convex and the limit $\lim_{x\to 1^-} \eta'(x)$ is finite. Therefore, the condition for $R=-H_{\eta}\in \Reg$ is steepness of $\eta$ at zero:
 \begin{equation}\label{eta0}
\lim_{x\to 0^+} \eta'(x)=\infty.
\end{equation}
Hence, $R=-H_{\eta}\in\Reg$  if and only if $\eta$ satisfies \eqref{eta0}.
\end{example}
Several well-known regularizers are given by (negative) Arimoto entropies satisfying \eqref{eta0}. For instance, the Shannon entropy from Example \ref{SE}  is an Arimoto entropy
 for $\eta(x)=-x\log x$, as well as log-barrier regularizer obtained from $\eta(x)=\log x$.  Another widely used (especially in statistical physics) example of Arimoto entropy is the Havrda-Charv\'{a}t-Tsallis entropy.\footnote{This entropy (called also entropy of degree $q$) was first introduced by Havrda and Charv\'{a}t \cite{HC} and used to bound probability of error for testing multiple hypotheses. In statistical physics it is known as Tsallis entropy, referring to \cite{Tsallis}.}

\begin{example}[Havrda-Charv\'{a}t-Tsallis entropies] \label{HCTent}
The Havrda-Charv\'{a}t-Tsallis entropy for $q\in (0,\infty)$ is defined as

\begin{equation} \label{Hq} H_{q}(x,y)=\begin{cases}  \frac{1}{1-q}(x^q+y^q-1) & \text{for } q\neq 1\\ H_S(x,y) & \text{for } q=1 \end{cases}. \end{equation}

$H_q$ is an Arimoto entropy for $\eta(x)= \frac{1}{1-q}\left(x^q-\frac 12\right)$, satisfying \eqref{eta0} for $0<q<1$. If $R(x,y)=-H_q(x,y)$ then

\[
r(x)=R(x,1-x)= \frac{1}{q-1}(x^q+(1-x)^q-1)\;\text{and}\; r'(x)=\frac{q}{q-1}\left(x^{q-1}-(1-x)^{q-1}\right),
\]
 and $r\in \Reg$ for $q\in (0,1]$.

For $q>1$ the Havrda-Charv\'{a}t-Tsallis entropy does not satisfy \eqref{eta0} and, consequently, the regularizer $R$ emerging from the Havrda-Charv\'{a}t-Tsallis entropy does not belong to $\Reg$.
Standard non-example is Euclidean norm, which we get from \eqref{Hq} when  $q=2$. Then \[r(x)=R(x,1-x)=-H_{2}(x,1-x)=  x^2+(1-x)^2-1\] and as $\lim_{x\to 0^+}r'(x)=-2$, $R$ doesn't belong to $\Reg$.
\end{example}

 Evidently there exist functions which are not Arimoto entropies but also generate regularizers that belong to $\Reg$, one of them being the R\'{e}nyi entropy of order $q<1$.

\begin{example}[R\'{e}nyi entropies]  The Shannon entropy represents an expected mean of individual informations of the form $I_k=-\log p_k$. R\'{e}nyi \cite{Renyi} introduced alternative information measures, namely generalized means $g^{-1}(\sum p_kg(I_k))$, where $g$ is a continuous, strictly monotone function.
Then, the R\'{e}nyi entropy of order $q\neq 1$ correspond to $g(x)=\exp ((1-q)x)$, namely:

\[
H_{q}^R(x,y)= \begin{cases} \frac{1}{1-q} \log \left(x^q+y^q \right), & \text{for } q\neq 1\\
H_S(x,y), & \text{for } q=1
\end{cases}.
\]
As the variables $x$ and $y$ are not separable, this is not an Arimoto entropy. However, for $R(x,y)=-H_q^R(x,y)$, $R\in \mathcal{C}^2((0,1)^2)$ and $R(y,x)=R(x,y)$.
Moreover,
\[
r(x)=R(x,1-x)=-H_q^R(x,1-x)=\frac{1}{q-1} \log \left( x^q+(1-x)^q \right)
\]
and

\[
r'(x)=\frac{q}{q-1}\cdot \frac{x^{q-1}-(1-x)^{q-1}}{x^q+(1-x)^{q}}.
\]

Thus, for $q\in (0,1)$ we know that $r''(x)>0$ on $(0,1)$ and $\lim_{x\to 0^+}r'(x)=-\infty$. Because $H_1^R=H_S$ we infer that  $R\in \Reg$ for $q \in (0,1]$.
\end{example}

\section{Regularity of log-barrier dynamics}
\label{s:regularity}
To understand better the phenomenon discussed in Section \ref{s:experimental}, let us investigate regularity of $f_{a,b}$. 
Nice properties of interval maps are guaranteed by the negative Schwarzian derivative. Let us recall
that the Schwarzian derivative of $f$ is given by the formula
\[
Sf=\frac{f'''}{f'}-\frac32\left(\frac{f''}{f'}\right)^2.
\]
A ``metatheorem'' states that almost all natural noninvertible
interval maps have negative Schwarzian derivative. Note that, by Lemma \ref{lem33}.\ref{commute}, if $a\le
-\Psi'(b)$ then  $f_{a,b}$ is a homeomorphism, so we should not expect negative
Schwarzian derivative for that case.
For maps with negative Schwarzian derivative each attracting or
neutral periodic orbit has a critical point in its immediate basin of
attraction. Thus, if we show that the Schwarzian derivative is negative, then we will know that all periodic orbits can be find by studying behavior of critical points of $f_{a,b}$.
Therefore, we want to show that $Sf_{a,b}<0$ for sufficiently large $a$ for $f_{a,b}$ determined by log-barrier regularizer.

In general, computation of Schwarzian derivative may be very complicated. However, there is a useful formula
\begin{equation}\label{Sdf}
 S(h \circ f) = (f')^2 \left( (Sh) \circ f \right) + Sf.
\end{equation}

The function $f_{a,b}$ is given by \eqref{eq:dyn2}.
Consider
\[
 g(x) := (\Psi \circ f_{a,b})(x) = \Psi(x) + a(x-b).
\]
By \eqref{Sdf} we have that \[Sg = (f_{a,b}')^2 \left( (S\Psi) \circ f_{a,b} \right) + Sf_{a,b}.\] At the same time \[Sg(x) = S(\Psi(x) + a(x-b)).\] Therefore,
\begin{equation}\label{Sdfab}
(f_{a,b}'(x))^2 \left( (S\Psi) \circ f_{a,b}(x) \right) + Sf_{a,b}(x) = S(\Psi(x) + a(x-b)).
\end{equation}


Direct computations yield
\[ S\Psi(x) = \frac{6}{\left( x^2 + (1-x)^2 \right)^2} > 0\]
\begin{minipage}{\textwidth}
and
\[
 S(\Psi(x) + a(x-b)) = \frac{6\left[ 1 - a(x^4 + (1-x)^4) \right]}{\left[ x^2 + (1-x)^2 - ax^2(1-x)^2 \right]^2}.
\]
Observe that $x^4 + (1-x)^4 \geqslant \frac{1}{8}$ for all $x \in [0,1]$. Thus $1 - a(x^4 + (1-x)^4) \leqslant 1 - \frac{a}{8}$, and 
\begin{equation}\label{SdLog_barrier2}
 S(\Psi(x) + a(x-b)) < 0 \quad \text{for } a > 8.
\end{equation}
Therefore, 
 $Sf_{a,b} < 0$ for $a > 8$.
Moreover, \[\max_{b\in [0,1]}\Psi'(b)=\Psi'(1/2)=-8.\] Thus,  \[ Sf_{a,b}(x)<0 \text{ for all } a>-\Psi'(b)\geq 8.\]

\end{minipage}



\end{document}